\documentclass[11pt,english]{article}
\usepackage[T1]{fontenc}
\usepackage[latin9]{inputenc}
\usepackage{geometry}
\usepackage{amsmath}
\usepackage{amsthm}
\usepackage{amssymb}
\usepackage{stackrel}
\usepackage[authoryear]{natbib}

\makeatletter
\theoremstyle{definition}
\newtheorem{defn}{\protect\definitionname}
\theoremstyle{plain}
\newtheorem{thm}{\protect\theoremname}
\theoremstyle{plain}
\newtheorem{cor}{\protect\corollaryname}
\theoremstyle{plain}
\newtheorem{prop}{\protect\propositionname}
\theoremstyle{plain}
\newtheorem{fact}{\protect\factname}
\theoremstyle{plain}
\newtheorem{lem}{\protect\lemmaname}

\@ifundefined{date}{}{\date{}}
\usepackage{babel}

\usepackage{babel}
\usepackage{listings}

  \providecommand{\corollaryname}{Corollary}
  \providecommand{\definitionname}{Definition}
  
  \providecommand{\lemmaname}{Lemma}
\providecommand{\theoremname}{Theorem}

\usepackage{babel}

  \providecommand{\definitionname}{Definition}
  
  \providecommand{\lemmaname}{Lemma}
\providecommand{\corollaryname}{Corollary}
\providecommand{\theoremname}{Theorem}

\usepackage{babel}
\addto\captionsaustralian{}
  \addto\captionsaustralian{\renewcommand{\definitionname}{Definition}}
  \addto\captionsaustralian{}
  \addto\captionsaustralian{\renewcommand{\lemmaname}{Lemma}}
  \addto\captionsenglish{}
  \addto\captionsenglish{\renewcommand{\definitionname}{Definition}}
  \addto\captionsenglish{}
  \addto\captionsenglish{\renewcommand{\lemmaname}{Lemma}}
  
  \providecommand{\definitionname}{Definition}
  
  \providecommand{\lemmaname}{Lemma}
\addto\captionsaustralian{\renewcommand{\corollaryname}{Corollary}}
\addto\captionsaustralian{\renewcommand{\theoremname}{Theorem}}
\addto\captionsenglish{\renewcommand{\corollaryname}{Corollary}}
\addto\captionsenglish{\renewcommand{\theoremname}{Theorem}}
\providecommand{\corollaryname}{Corollary}
\providecommand{\theoremname}{Theorem}

\usepackage{babel}
\addto\captionsaustralian{}
  \addto\captionsaustralian{\renewcommand{\definitionname}{Definition}}
  \addto\captionsaustralian{}
  \addto\captionsaustralian{\renewcommand{\lemmaname}{Lemma}}
  \addto\captionsenglish{}
  \addto\captionsenglish{\renewcommand{\definitionname}{Definition}}
  \addto\captionsenglish{}
  \addto\captionsenglish{\renewcommand{\lemmaname}{Lemma}}
  
  \providecommand{\definitionname}{Definition}
  
  \providecommand{\lemmaname}{Lemma}
\addto\captionsaustralian{\renewcommand{\corollaryname}{Corollary}}
\addto\captionsaustralian{\renewcommand{\theoremname}{Theorem}}
\addto\captionsenglish{\renewcommand{\corollaryname}{Corollary}}
\addto\captionsenglish{\renewcommand{\theoremname}{Theorem}}
\providecommand{\corollaryname}{Corollary}
\providecommand{\theoremname}{Theorem}

\usepackage{babel}
\addto\captionsaustralian{}
  \addto\captionsaustralian{\renewcommand{\definitionname}{Definition}}
  \addto\captionsaustralian{}
  \addto\captionsaustralian{\renewcommand{\lemmaname}{Lemma}}
  \addto\captionsaustralian{\renewcommand{\propositionname}{Proposition}}
  \addto\captionsenglish{}
  \addto\captionsenglish{\renewcommand{\definitionname}{Definition}}
  \addto\captionsenglish{}
  \addto\captionsenglish{\renewcommand{\lemmaname}{Lemma}}
  \addto\captionsenglish{\renewcommand{\propositionname}{Proposition}}
  
  \providecommand{\definitionname}{Definition}
  
  \providecommand{\lemmaname}{Lemma}
  \providecommand{\propositionname}{Proposition}
\addto\captionsaustralian{\renewcommand{\corollaryname}{Corollary}}
\addto\captionsaustralian{\renewcommand{\theoremname}{Theorem}}
\addto\captionsenglish{\renewcommand{\corollaryname}{Corollary}}
\addto\captionsenglish{\renewcommand{\theoremname}{Theorem}}
\providecommand{\corollaryname}{Corollary}
\providecommand{\theoremname}{Theorem}

\usepackage{babel}
\addto\captionsaustralian{}
  \addto\captionsaustralian{\renewcommand{\definitionname}{Definition}}
  \addto\captionsaustralian{}
  \addto\captionsaustralian{\renewcommand{\lemmaname}{Lemma}}
  \addto\captionsaustralian{\renewcommand{\propositionname}{Proposition}}
  \addto\captionsaustralian{}
  \addto\captionsenglish{}
  \addto\captionsenglish{\renewcommand{\definitionname}{Definition}}
  \addto\captionsenglish{}
  \addto\captionsenglish{\renewcommand{\lemmaname}{Lemma}}
  \addto\captionsenglish{\renewcommand{\propositionname}{Proposition}}
  \addto\captionsenglish{}
  
  \providecommand{\definitionname}{Definition}
  
  \providecommand{\lemmaname}{Lemma}
  \providecommand{\propositionname}{Proposition}
  
\addto\captionsaustralian{\renewcommand{\corollaryname}{Corollary}}
\addto\captionsaustralian{\renewcommand{\theoremname}{Theorem}}
\addto\captionsenglish{\renewcommand{\corollaryname}{Corollary}}
\addto\captionsenglish{\renewcommand{\theoremname}{Theorem}}
\providecommand{\corollaryname}{Corollary}
\providecommand{\theoremname}{Theorem}

\usepackage{babel}
\addto\captionsaustralian{}
  \addto\captionsaustralian{\renewcommand{\definitionname}{Definition}}
  \addto\captionsaustralian{}
  \addto\captionsaustralian{\renewcommand{\lemmaname}{Lemma}}
  \addto\captionsaustralian{\renewcommand{\propositionname}{Proposition}}
  \addto\captionsaustralian{}
  \addto\captionsenglish{}
  \addto\captionsenglish{\renewcommand{\definitionname}{Definition}}
  \addto\captionsenglish{}
  \addto\captionsenglish{\renewcommand{\lemmaname}{Lemma}}
  \addto\captionsenglish{\renewcommand{\propositionname}{Proposition}}
  \addto\captionsenglish{}
  
  \providecommand{\definitionname}{Definition}
  
  \providecommand{\lemmaname}{Lemma}
  \providecommand{\propositionname}{Proposition}
  
\addto\captionsaustralian{\renewcommand{\corollaryname}{Corollary}}
\addto\captionsaustralian{\renewcommand{\theoremname}{Theorem}}
\addto\captionsenglish{\renewcommand{\corollaryname}{Corollary}}
\addto\captionsenglish{\renewcommand{\theoremname}{Theorem}}
\providecommand{\corollaryname}{Corollary}
\providecommand{\theoremname}{Theorem}

\setcitestyle{round}
\usepackage[hang,flushmargin]{footmisc}
\usepackage[charter]{mathdesign}

\makeatother

\usepackage{babel}
\providecommand{\corollaryname}{Corollary}
\providecommand{\definitionname}{Definition}
\providecommand{\factname}{Fact}
\providecommand{\lemmaname}{Lemma}
\providecommand{\propositionname}{Proposition}
\providecommand{\theoremname}{Theorem}

\begin{document}
\begin{titlepage}
\title{\textbf{On uniform boundedness of sequential social learning}}
\author{Itay Kavaler$^{a}$$^{,}$\thanks{\ Corresponding author.\protect \\
\indent \ \ \ E-mail addresses: itayk@campus.technion.ac.il (I.Kavaler).}}
\date{$^{a}$Davidson Faculty of Industrial Engineering and Management,
Technion, Haifa 3200003, Israel \\
}
\maketitle
\begin{abstract}
In the classical herding model, asymptotic learning refers to situations
where individuals eventually take the correct action regardless of
their private information. Classical results identify classes of information
structures for which such learning occurs. Recent papers have argued
that typically, even when asymptotic learning occurs, it takes a very
long time. In this paper related questions are referred. We studiy
whether there is a natural family of information structures for which
the time it takes until individuals learn is uniformly bounded from
above. Indeed, we propose a simple bi-parametric criterion that defines
the information structure, and on top of that compute the time by
which individuals learn (with high probability) for any pair of parameters.
Namely, we identify a family of information structures where individuals
learn uniformly fast. 

The underlying technical tool we deploy is a uniform convergence result
on a newly introduced class of `weakly active\textquoteright{} supermartingales.
This result extends an earlier result of \citet{Fudenberg-Levine-1992}
on active supermartingales.

\vspace{1cm}


\smallskip{}

\end{abstract}
\vspace{0cm}

\section{Introduction}

Oftentimes, when individuals make decisions, they rely on their private
information as well as on previous decisions that other individuals
have made. Observational learning models have been studied on purpose
to formalize and analyze this insightful behavior. The paradigmatic
setup of the model (e.g. \citet{Sushil-1992}) includes an infinite
sequence of individuals who choose actions with the objective to match
the unknown (binary) state of the world. Each individual receives
a private information which is iid conditioned on the underlying state,
and in addition to that, observes the entire sequence of earlier actions.

\citet{Lones-Smith-2000} study the scenario where the private information
is arbitrarily strong (unbounded in their jargon). Informally, each
individual may obtain (with positive probability) a private information
in favor to one state that will overturn any evidence in favor to
the contrary one. When the private information is unbounded and all
actions are publicly observed, \citet{Lones-Smith-2000} have shown
that asymptotic learning occurs with probability one: the public belief
converges to the Dirac measure on the realized state and individuals
eventually take the correct action.\footnote{The rationale behind this learning is that, whenever a sequence of
incorrect actions is unfolding there exist individuals who will eventually
put an end to it regardless of their belief assigned to the incorrect
state.}

Although asymptotic learning has been established for unbounded signals,
limited attention has been paid in the literature to speed at which
asymptotic learning converges and the rate at which the public belief
approaches one. Recent evidence suggests that learning is typically
`slow' (see for example the recent paper by \citet{Tamuz-2017}).
Motivated by this we identify the structural properties of the information
structure (where the last refers to a pair of distributions over posterior
beliefs, one for each state of nature\textbf{)} for which the aforementioned
convergence not only happens \textquoteleft quickly\textquoteright{}
but in a uniform rate. To that end, we revisit the notion of informativeness.

Private information must not perfectly reveal the state of the world.
This ensures that any (set of) information  which is possible in regard
to one state, is possible in regard to the other state. In addition
to that, private information is  assumed to rule out trivialities
in the sense that some (set of) information  is measured differently
with respect to the states. As a supplemental assumption, the last
property is refined in the following manner.

Recall that private information is naturally associated with a pair
of distributions corresponding to the set of states of the world.
Informally, we focus on a certain classes of information structures
(pairs of posterior beliefs) that admit a uniform rate of convergence.
Each class is bi-parametric, defined by a pair of parameters $\psi$
and $\nu$ in $(0,1)$, such that the absolute value of the point-wise
distances, corresponding to each pair of distributions in that class,
decreases at a uniform bounded rate  which is merely determined by
these parameters. In particular, $\psi$ plays as an initial constraint
whereas $\nu$ determines the bounded rate parameter of convergence.

At each stage consider the ratio between the public belief over one
state and that belief over the other state. This process is shown
to converge `quickly' as the distance between the corresponding pairs
decreases fast. Our main result shows that the ratio processes induced
by all pairs corresponding to the same bi-parameter class, represented
by some $\psi$ and $\nu,$ converge 'fast' and at a uniformly bounded
rate which is solely determined by these parameters. In addition,
we show that these classes are tight in the sense that once we relax
the bi-parametric constrain, the ratio process may converge very slow.

It is well known (\citet{Lones-Smith-2000}) that the ratio process
is a state conditional martingale (and hence it is a suprmartingale)
which converges to the correct cascade set with probability one. We
introduce a new class of supermartingales, called {\em weakly active},
where one of which process admits a `high' rate of convergence which
is bounded from below by $\frac{\psi}{(t+1)^{v}},$ for any given
$\psi$ and $\nu$ at each time $t.$

\vspace{0cm}

Our prevailing results are mathematical: we show that these classes
of supermartingales, identified by $\psi$ and $\nu$, converge uniformly.
Specifically, given a precision value $\epsilon$, there exits a constant
time $K$, which entirely depends on $\psi$, $\nu$ and $\epsilon$,
such that with high $1-\epsilon$ probability the entire processes
in the associated class (uniformly) fall and remain below $\epsilon$
from time $K$ onward.

In addition, we show that there is a one to one correspondence between
the classes of `weakly active' supermartingales and the classes of
bi-parametric private information structures. In particular we show
that any ratio process which is uniquely induced from a bi-parametric
class, identified $\psi$ and $\nu$, is a weakly active supermatingale
defined by the same parameters.

Finally, as learning is merely asymptotically approachable, the wrong
action may be selected with a positive probability at any finite time.
In particular, the sequence of actions may not settle on the correct
action from some finite time onward and more rigorously, it may require
an unreasonable amount of time to put an end of a sequence of incorrect
actions. \citet{Dinah-Rosenberg-2017} consider the following two
criteria: One criterion looks at how long is needed until an individual
takes the correct action, the other one measures the total number
of incorrect actions. \citet{Dinah-Rosenberg-2017} study conditions
over the information structures for which the expectations of these
two variables is finite. 

While their (necessary) condition for efficiency involves a finite
indefinite integral regarding the unconditional private information's
distribution, this criterion is surprisingly captured by the informativeness
approach taken here in the following rigorous form; We show that the
aforementioned expected times, corresponding to any bi-parametric
class of  information identified by $\psi$ and $\nu$, are not only
finite but uniformly bounded by a constant which solely depends on
$\psi$ and $\nu$ and is independent of the distribution\textquoteright s
specification in that class.

\subsection{Results}

We introduce a new condition over the set of information structures.
We show that the stochastic ratio processes induced by an information
structure satisfying the condition are supermartingales which are
shown to admit a uniform bounded rate of convergence. Additionally,
once the condition is relaxed, then either a `fast' convergence occurs
or the associated private information distributions are sufficiently
`close', yielding that less information about the true state of the
world is obtained. In particular, under the proposed bi-parametric
constrain, we show that the expected times of the first correct action
as well as the total number of incorrect actions are finite and uniformly
bounded.

The substantial results are crucially relied on our preeminent mathematical
contribution. A new class of weakly active supermartingales is introduced,
linked and shown  to admit a uniform rate of convergence.

\subsection{Related literature}

\citet{Chamley-2004} provides an estimation for the the public belief
unfolding in case for which the associated private information distributions
admit 'fat tails' in his jargon. He also reports some numerical evidence
in the Gaussian case using a computer simulation. \citet{Tamuz-2017}
argue that the rate growth of the result in ratio process, associated
with individuals who observe actions, is sub-exponentially close.

\citet{Xavier-Vives-1993}, studies the speed of sequential learning
in a model with actions chosen from a continuum where individuals
observe an information with error term regarding their predecessors'
actions. He similarly shows that learning is significantly slower
than in the threshold case (for further overview of this literature
please refer to \citet{Vives-book}, Chapter 6). For specific  information
structure, some bounds on the probability that an individual takes
the incorrect action have been established. Under the assumption that
private information is  uniformly distributed, \citet{Ilan-Lobel-2009}
establish upper bounds on this probability, where \citet{Tamuz-2017}
provide a lower bound, under the assumption that private information
is  normally distributed, and all previous actions are publicly observed.

Much of the literature has focused on the long-term outcomes of learning.
As aforestated, when private information is unbounded, the individuals
eventually choose the correct action with probability one. Hence,
a natural question one may arise is how long does it take for that
to happen and study its expectation.

\citet{Lones-Smith-1996}, an early version of \citet{Lones-Smith-2000},
addressed this issue showing that the expected time to learn is infinite
for some private information structures. They further conjecture that
this result is extended to an arbitrary structure, which on the contrary,
is shown to be false. \citet{Dinah-Rosenberg-2017} have studied related
questions. In particular, they identify a simple condition on the
unconditional private information distributions for which the expected
times until the first correct action, as well as the number of incorrect
actions, are finite. Furthermore, \citet{Tamuz-2017} poses a different
assumption on the conditional private information distributions' tails
for which the expected time to learn is finite.

Those two approaches are distilled here collectively as the aforementioned
expected times are captured under our new informativeness notion.

\vspace{0cm}

\section{Model\label{sec:Model}}

The general framework includes the Nature who chooses a binary state
of the world within $S=\{L,H\}.$ For notional simplicity both states
are assumed to be drawn likely. A sequence of short-lived individuals
make decisions in turn. In round $k$, an individual $k$ chooses
an action from the set $A=\{1,...,m,...,M\}$ with the objective of
matching the underlying unknown state. Significant savings in computational
cost, we consider the case for which $M=2.$

Let $U:S\times A\longrightarrow\mathbb{R}$ be a payoff function.
In each state $s\in\{L,H\}$ an individual earns a payoff $U(s,m)$
from the action $m$ and seeks to maximize his payoff. It is assumed
that $U(H,m)$ is increasing in $m$ whereas $U(L,m)$ is decreasing,
and in particular no action is optimal over another at one belief.
This will avoid that no pair of actions provides identical payoffs
in all states. It will further ensure that each individual has two
extreme actions, each is strictly optimal in some state. Particularly,
action $2$ is optimal for $H$ where action $1$ is optimal for $L.$

\vspace{0cm}

{\em Private belief}: Each individual receives a random signal
about the true state of the world and then uses Bayes rule to compute
his private belief.\footnote{We adopt the common literature terminology and refer to that belief
as the private information an individual assigns to the states, given
his own information.} Conditional on the states $H,L,$ a private belief (rather then a
signal) is drawn iid according to some cdfs: $F^{H},F^{L},$ respectively.
It is further assumed that $F^{L},F^{H}$ have a common support, say
$supp(F),$ such that the convex hull $co(supp(F))\equiv[\underset{-}{b},\stackrel{-}{b}]\subseteq[0,1]$
for some $0<\underset{-}{b}<\stackrel{-}{b}<1$.

Two additional extreme restrictions must be posed on $F^{L}$ and
$F^{H};$ No private belief perfectly reveals the true state of the
world. Formally, $F^{L}\sim F^{H},$ or equivalently, the Radon-Nikodym
$f=\frac{dF^{L}}{dF^{H}}$ exists and is finite at any point in $supp(F).$\footnote{$F^{H}\sim F^{L}$ implies that $F^{H}(p)>0$ if and only if $F^{L}(p)>0$
at any point $p$ in $supp(F).$} For the second restriction it is assumed that some signals are informative;
for avoiding triviality $F^{L}$ and $F^{H}$ must not coincide. Formally,
 $\underset{-}{b}<\frac{1}{2}<\stackrel{-}{b},$ this rules out that
$f=1$ a.s.. A private belief is said to be bounded if $0<\underset{-}{b}<\stackrel{-}{b}<1,$
and it is called unbounded whenever $co(supp(F))\equiv[\underset{-}{b},\stackrel{-}{b}]=[0,1].$

\vspace{0cm}

{\em Posterior belief}: Given a posterior belief $r\in[0,1]$ that
the true state is $H,$ the expected payoff for choosing action $m$
is $rU(H,m)+(1-r)U(L,m)$. As each individual maximizes his expected
payoff, $U$ and $r$ naturally induce a unique thresholds:

\noindent 
\begin{equation}
0=r_{0}<r_{1}<r_{2}=1\label{eq:trheshods in r}
\end{equation}

\noindent so that each interval $I_{m}:=[r_{m-1},r_{m}]$ is associated
with an optimal action $m$ that the expected payoff of choosing action
$m$ is maximal. Additionally, in equilibrium, each individual can
compute the probability that the state is $H$ after any history of
choices of all previous individuals. This probability is referred
to the {\em public belief} and is denoted by $q$. Applying Bayes
rule for a given public belief $q$ and a private belief $p$, the
posterior belief $r$ that the state is $H$ satisfies

\begin{equation}
r=r(p,q)=\frac{pq}{pq+(1-p)(1-q)}.\label{eq:threreshold in P}
\end{equation}

{\em Game timing}: For each $k$, individual $k$ observes his
private signal and all actions of previous individuals, using Bayes
rule he then updates his posterior belief and chooses an action which
maximizes his payoff accordingly.

\vspace{0cm}

{\em Threshold}: Since the right hand side of \eqref{eq:threreshold in P}
is increasing in $p,$ there are private belief thresholds

\noindent 
\begin{equation}
0=p_{0}(q)\leq p_{1}(q)\leq p_{2}(q)=1,\label{eq:threshold in p(q)}
\end{equation}

\noindent corresponding to inequalities \eqref{eq:trheshods in r}
so that, given $q,$ an individual chooses the action $m$ if and
only if his private belief $p\in(p_{m-1}(q),\ p_{m}(q)].$ Equality
\eqref{eq:threreshold in P} further implies that each threshold $p_{m}(q)$
is decreasing in $q.$ 

\vspace{0cm}

{\em Cascade set}: Each action $m$ is associated with a set of
public beliefs $J_{m}=\{q|\ supp(F)\subseteq(p_{m-1}(q),\ p_{m}(q)]\}.$
Hence, an individual takes an action $m$ almost surely whenever $q\in int(J_{m})$.

\subsection{Asymptotic learning}

\noindent Let $q_{k}$ be the public belief that the state is $H$
which, using Bayes rule, is updated accordingly after individual $k$
chooses an action. Conceptually, since we are interested in the conditional
stochastic properties that the state is $H,$ it is more convenience
to consider the {\em public likelihood ratio} $l_{k}\equiv\frac{1-q_{k}}{q_{k}}$
that the state is $L$ versus $H.$ 

The following analogues can then be inferred; for each $q$ the public
belief thresholds are $\bar{p}_{m}(l:=\frac{(1-q)}{q})=p_{m}(q);$
for each posterior belief $r\in I_{m}$ if and only if $\frac{(1-r)}{r}\in\bar{I}_{m}:=[\frac{1-r_{m-1}}{r_{m-1}},\frac{1-r_{m}}{r_{m}}];$
and for each cascade set $q\in J_{m}$ if and only if $\frac{(1-q)}{q}\in\bar{J}_{m}$.

By (\citet{Lones-Smith-2000}, Lemma 2) we know that, with unbounded
belief, $\bar{J}_{1}=\{\infty\}$ and $\bar{J}_{2}=\{0\},$ where
all the other cascade sets are empty. Moreover, conditional on the
state $H$, $l_{k}$ is a converges martingale with $l_{k}\rightarrow\bar{J}_{2}=\{0\},$
and hence fully correct asymptotic learning  occurs almost surely
(for further details please refer to \citet{Lones-Smith-2000}, Lemma
3).

Since, conditional on $l_{k},$ an individual $k$ chooses the action
$m$ if and only if his private belief $p_{k}\in[\bar{p}_{m-1},\bar{p}_{m}],$
it follows that the corresponding likelihood public belief at time
$k+1,$ $l_{k+1},$ is updated according to the likelihood posterior
belief $\frac{(1-r_{k})}{r_{k}}.$ Hence, as each state is equally
likely, it is natural to describe the stochastic process $\{l_{k}\}_{k\geq0}$
iteratively via

\begin{equation}
l_{k+1}=l_{k}\frac{\rho(\cdot|l_{k},L)}{\rho(\cdot|l_{k},H)},\ l_{0}\equiv1\label{eq:iteratively of ln}
\end{equation}

\noindent where the transition

\begin{equation}
\rho(m|l,s)=F^{s}(\bar{p}_{m}(l))-F^{s}(\bar{p}_{m-1}(l)),\label{eq:transition}
\end{equation}

\noindent is the probability that an individual takes the action $m$
given $l$ where the true state is $s\in\{H,L\}.$ Consequently, the
cascade set $\bar{J}_{m}$ is the interval of the likelihoods $l$
such that 
\begin{equation}
\rho(m|l,H)=\rho(m|l,L)=1,\label{eq:transitions converfes to 1}
\end{equation}

\noindent and further, as the likelihood posterior belief jumps above
one whenever an individual follows an incorrect action, it can be
inferred that for every partial sequence of actions $a_{1},...,a_{k+1},$

\noindent 
\[
l_{k+1}(a_{1},...,a_{k+1})=l_{k}(a_{1},...,a_{k})\frac{\rho(a_{k+1}|l_{k},L)}{\rho(a_{k+1}l_{k},H)}>1\text{if and only if }a_{k+1}=1.
\]

\section{Uniform estimation}

Recall that a pair of cdfs, $(F^{L},F^{H}),$ together with an initial
probability $P(H)=P(L)=\frac{1}{2}$ induce a unique stochastic process
$l_{k}$ which satisfies equation \eqref{eq:iteratively of ln}. Additionally,
recall that when $F^{L},F^{H}$ are unbounded, $l_{k}$ converges
almost surly to the cascade set $\bar{J}_{2}=\{0\}$ corresponding
to the true action $2$. As the transitions converge to one, by \eqref{eq:transitions converfes to 1},
the process $l_{k}$ decays `slowly' to zero (For suplemental details
please refer to \citet{Tamuz-2017}, Theorems 1,2). Hence, we are
motivated to establish a class of processes which will not only dismantle,
in some extent, this `sluggishness' convergence in moderation, but
admit a uniform rate of converges. To capture this idea let us consider
the following class of unbounded private beliefs cdfs.

Each pair ($F^{L},F^{H})$ is equipped with the following iterative
deterministic sequence: $\tilde{l}_{k+1}=\tilde{l}_{k}\frac{\rho(2|\tilde{l}_{k},L)}{\rho(2\tilde{l}_{k},H)}$
started at $\tilde{l}_{0}=1.$ In words, for each $k\geq0$ the value
of $\tilde{l}_{k}$ is associated with a particular partial history
of choices in which all individuals, from day one to $k,$ pick the
correct action $2$. Hence, $(\tilde{l}_{k})_{k\geq0}$ is monotone
decreasing to zero. In addition, let us denote by $\Delta(p)=F^{L}(p)-F^{H}(p)$
to be the distance between two corresponding $\text{cdf}'s$ at $p.$
\begin{defn}
\label{DEF: informative CDFs}Let $\psi,\nu\in(0,1).$ A pair of unbounded
private beliefs cdf's, ($F^{L},F^{H}),$ is {\em$(\psi,\nu)$ -
informative} if

\begin{equation}
\Delta(\bar{p}_{1}(\tilde{l}_{k}))>\frac{\psi}{(k+1)^{\nu}}\label{eq: condition (psi,v) - informative}
\end{equation}

\noindent for all $k\geq0.$
\end{defn}
In order to justify the information structure in which the private
belief $p$ (rather than private signal) is iid conditioned on the
underlying state, the following correspondence must arise from Bayesian
updating: $f(p)=(1-p)/p$ almost surely in $(0,1)$. This immediately
implies that $\Delta(p)$ is non-decreasing as $p<1/2$, and furthermore,
since it is strict whenever $p\in supp(F)\setminus\{1/2\},$ it follow
that $F^{H}<F^{L}$, except when both terms are $0$ or $1$ (for
further details please refer to Lemma A.1 in \citet{Lones-Smith-2000}).

Therefore, the $(\psi,\nu)$ - informativeness assumption entails
the following restriction: whereas $\Delta(p)$ weakly decreases as
$p$ tends to zero, this distance, on a particular sample sequence
$\tilde{l}_{k},$ cannot decay too slowly but at rate of at least
$\frac{\psi}{(k+1)^{\nu}}.$ A pair of cdfs  which complies with this
restriction admits thick (weakly monotonic) tails.

\vspace{0cm}

The next theorem provides a characterization of unbounded private
beliefs. Notably, the construction includes a wide class of distributions
which complies with Definition \eqref{DEF: informative CDFs}. In
particular, it is argued that when the distance between the tails
of $F^{L}$ and $F^{H}$ are sufficiently thick then the (conditional)
probability that strong signals vanishes slowly enough is high. Hereinafter,
$P^{H}$ is referred to the conditional probability which by $F^{H}.$
\begin{thm}
\textup{\label{Th: existance of inforamtive-CDF}For all $\psi,\nu\in(0,1)$
there exists a pair $(F^{L},F^{H})$ which is $(\psi,\nu)$ - informative.}
\end{thm}
The proof of Theorem \ref{Th: existance of inforamtive-CDF} is supplemented
to Appendix \ref{Appendix A sec:Selected-proofs}. Notably, the constructive
proof generates a broad classes of unbounded private beliefs which
comply with inequality \eqref{eq: condition (psi,v) - informative}. 

The next uniform theorem states that each bi-parametric class of $(\psi,\nu)$
- informative pairs is identified uniformly in finite time. In particular,
the theorem establishes a uniform bound on the rate of convergence
corresponding to any process $l_{k}$ which is induced by $(\psi,\nu)$
- informative private beliefs. The theorem asserts that all such processes
decay at rate which is uniformly bounded.
\begin{thm}
\textup{\label{Th: uniform bound for (psi,nu) informative}Let $\epsilon,\psi,\nu\in(0,1).$
Then for all $\underline{L}<1$ there is a finite time $K=K(\psi,\nu,\epsilon\text{,\ensuremath{\underline{L}}})$
such that for all pairs $(F^{L},F^{H})$ which are $(\psi,\nu)$ -
informative there is a $P^{H}$ - probability of at least $(1-\epsilon)$
that}

\textup{
\[
l_{k}<\underline{L}
\]
}

\noindent \textup{for all $k>K.$}
\end{thm}
\vspace{0cm}

In words, the uniformity theorem asserts that given $\psi,\nu$ and
a precision $\epsilon,$ there exists a constant time $K,$ which
depends only on these three parameters, such that after exactly first
$K$ guesses not only that the likelihood public belief process at
time $K+1$, $l_{K+1},$ falls below $\epsilon$ near the cascade
set but, due to the fact that $\bar{p}_{1}(l_{k})$ is decreasing
in $k,$ $\bar{p}_{1}(l_{K+1})$ is near zero and hence the $K+1$
individual is guaranteed that all the individuals thereafter will
choose the correct action $2$ with high probability of at least $\rho(2|l_{K+1},H).$
Consequently, in the later scenario, there is an upper uniform bound
on the time after which all individuals will choose the correct action
with at least that uniform probability.

\vspace{0cm}

The uniformity stated in Theorem \ref{Th: uniform bound for (psi,nu) informative}
implies the following corollary.
\begin{cor}
\textup{For all $\epsilon,\psi,\nu\in(0,1)$ there exists a finite
$K=K(\psi,\nu,\epsilon)$ such that for all pairs $(F^{L},F^{H})$
which are $(\psi,\nu)$ - informative, $P^{H}(\{a_{k}\}_{k\geq1}:a_{k}=2\text{\ for\ all}\ k>K)>1-\epsilon.$ }

In words, each bi-parametric class is characterized by an upper uniform
bound on the time after which all individuals will choose the correct
action.
\end{cor}
\vspace{0cm}

The next theorem analyzes the rate of converges of $l$ regardless
of the type of the private belief distributions which may be bounded
or unbounded. To this end, we shell first introduce the notion `close'
private beliefs.
\begin{defn}
A pair  ($F^{L},F^{H})$ is {\em$\epsilon$ - close} given $p$
if $\Delta(p)<\epsilon.$
\end{defn}
In particular, private beliefs are $\epsilon$ close given $\bar{p}_{1}(l)$
if the (uniquely) induced probabilities are sufficiently close for
all actions $m$, formally

\begin{equation}
\underset{m}{max}|\rho(m|l,H)-\rho(m|l,L)|<\epsilon.\label{DEF eq:psi,v - informative-1-2}
\end{equation}

In words, private beliefs are $\epsilon$ close given $\bar{p}_{1}(l)$
if it induces sufficiently close transitions for all actions $m.$
In the following theorem the assumption of $(\psi,\nu)$ - informativeness
is relaxed, more over, the theorem holds for all types of private
beliefs which might be either bounded or unbounded. In particular,
it is shown that the rate of converges does not depend on the specification
of the private beliefs but only on the distance between the corresponding
transitions.

\vspace{0cm}

\vspace{0cm}

\begin{thm}
\label{Th: relexing informativeness}For all $\epsilon,\psi,\nu\in(0,1)$
there exists a finite \textup{$K=K(\psi,\nu,\epsilon)$} such that
for all pairs $(F^{L},F^{H})$ and for all (sufficiently large) $n,$
there is a set of which the probability according to $P^{H}$ is at
least $(1-\epsilon)$ such that for any realization $\omega$ in that
set
\begin{enumerate}
\item Either $(F^{L},F^{H})$ is $\frac{\psi}{k^{\nu}}$ - close given $\bar{p}_{1}(l_{k}(\omega))$
in all, but $K$ periods $k$ in $\{1...n\}$ or
\item $l_{k}(\omega)<\epsilon(1+\epsilon)\ \ensuremath{\text{for all }}k>n.$
\end{enumerate}
\end{thm}
\vspace{0cm}

In words, as individual $n$ observes the whole previous actions as
well as the past associated transitions, Theorem \ref{Th: relexing informativeness}
asserts that, given an arbitrarily small $\epsilon>0,$ there exists
a finite uniform bound $K$, is independent of any pair of private
belief, such that if that individual faces more than $K$ periods
in which the corresponding transitions were sufficiently differed,
then, whenever the private beliefs are unbounded, not only that the
public belief process at time $n$, $l_{n},$ falls below $\epsilon$
near the cascade set but, the second part of theorem ensures us that
individual $n$ is guaranteed that all the individuals thereafter
will choose the same correct action $2$ with high probability of
at least $\rho(2|l_{n},H).$ Consequently, in the later scenario,
there is a universal upper uniform bound on the time after which all
individuals will choose the correct action with at least that probability.

When the private belief are $(\psi,\nu)$ - informative the first
part of Theorem \ref{Th: relexing informativeness} is ruled out and
hence Theorem \ref{Th: relexing informativeness} is coincided Theorem
\ref{Th: uniform bound for (psi,nu) informative}. As a result, with
 probability of at least $(1-\epsilon)$ the correct action will be
chosen after exactly $K$ periods.

By (\citet{Lones-Smith-2000}, Lemma 2) we know that, with bounded
belief, there exist $0<\underset{-}{l}<\stackrel{-}{l}<\infty$ such
that the non-empty cascade sets corresponding to actions $1$ and
$2$ are $\bar{J}_{1}=[\stackrel{-}{l},\infty]$ and $\bar{J}_{2}=[0,\underset{-}{l}],$
respectively, where all the other cascade sets are possibly empty. 

Consequently, since from Theorem \ref{Th: relexing informativeness}
we are provided with only two events which occur with high probability
$1-\epsilon$, and the event in which the process $l$ lies in $\bar{J}_{1}$
is obtained with a positive probability, we conclude that incorrect
`learning' (.i.e. herding) must develop very `slowly'  as the difference
between the transitions decays `fast' as less than $\frac{\psi}{k^{\nu}}$
in all but few periods $K$, whereas $l_{k}$ might converge `quickly'
to $\bar{J}_{2}.$ 

More importantly, the critical point is that the constant $K$ is
universal in the sense that it does not depend on any underlying pair's
specification. Consequently, we can remarkably conclude that the question
whether asymptotic learning has started at an arbitrarily time $n$
does not depend on the private belief's specification but only on
their natural distance.

\section{Supermartingale characterization}

Intuitively, a supermartingale is a process that decreases on average.
Let us further consider the following class of supermartingales called
active supermartingales. This notion has been first introduced in
Fundernberg and Levine (1992) who study reputations in infinitely
repeated games. Consider an abstract setting with a finite set $\Omega$
and probability measure $P$ in $\Delta(\Omega)$ equipped with a
natural filtration.

\vspace{0cm}

\begin{defn}
\label{Def: active super martingale}A non-negative supermartingale
$\tilde{L}:=\{\tilde{L}_{k}\}_{k=0}^{\infty}$ is {\em active} with
activity $\psi\in(0,1)$ under $P$ if

\[
P(\{\omega:|\frac{\tilde{L}_{k+1}(\omega)}{\tilde{L}_{k}(\omega)}-1|>\psi\}|\tilde{\omega}^{k})>\psi
\]

\noindent for $P$ - almost all histories $\tilde{\omega}^{k}$ such
that $\tilde{L}_{k}(\tilde{\omega})>0$.
\end{defn}
In word, a supermartingale  has activity $\psi$ if the probability
of a jump of size $\psi$ at time $k$ exceeds $\psi$ for almost
all histories. Note that $\tilde{L}$ being a supermartingale, is
weakly decreasing in expectations. Showing that it is active implies
that $\tilde{L}_{k+1}$ substantially goes up or down relative to
$\tilde{L}_{k}$ with probability bounded away from zero in each period.
\citet{Fudenberg-Levine-1992} showed the following remarkable result
\begin{thm}
\textup{\label{Th:  A.1 (F=000026L)}Let $\epsilon>0,\ \psi\in(0,1),$
and $l_{0}>0.$ Then, for all $\underset{-}{L}\in(0,l_{0})$ there
is a time $K<\infty$ such that}

\textup{
\[
P(\{\omega:\underset{k>K}{sup}\tilde{L}_{k}(\omega)\leq\underset{-}{L}\})\geq1-\epsilon
\]
}

\noindent \textup{for every active supermartingale $\text{\ensuremath{\tilde{L}} }$
with $\tilde{L}_{0}\equiv l_{0}$ and activity $\psi.$ }
\end{thm}
Theorem \ref{Th:  A.1 (F=000026L)} asserts that if $\tilde{L}$ is
an active supermartingale with activity $\psi$ then there is a fixed
time $K$ by which, with high probability, $\tilde{L}_{k}$ drops
below $\underset{-}{L}$ and remains below $\underset{-}{L}$ for
all future periods. It should be noted that the power of the theorem
stems from the fact that the bound, $K$, depends solely on the parameters
$\epsilon>0,\ \psi$ and $\underset{-}{L}$ , and is otherwise independent
of the underlying stochastic process $P$. 

\vspace{0cm}

Nevertheless, in the framework of social learning, non of the public
likelihood ratios $l_{k}$ exploits the active supermartingale property
for any $\psi\in(0,1),$ as equation \eqref{eq:transitions converfes to 1}
implies that its transitions converges to one. 

Hence, we are strongly motivated to extend Definition \ref{Def: active super martingale}
to a much broader class of supermatingales, called weakly active,
which includes $l_{k},$ and furthermore, as $l_{k}$ converges to
the cascade set $\{0\},$ extend Theorem \ref{Th:  A.1 (F=000026L)}
to achieve a uniform bound on the rate at which these processes converge. 

To this end, let us introduce the following new broad class of weakly
active supermartinglaes. This class is shown be associated with the
class of $(\psi,\nu)$- informative private beliefs.
\begin{defn}
Let $\psi,v\in(0,1)$. A non-negative supermartingale $\tilde{L}:=\{\tilde{L}_{k}\}_{k=0}^{\infty}$
is {\em weakly active}  with activity $\psi$ and rate $v$ under
if 
\begin{equation}
P(\{\omega:|\frac{\tilde{L}_{k+1}(\omega)}{\tilde{L}_{k}(\omega)}-1|>\frac{\psi}{(k+1)^{v}}\}|\tilde{\omega}^{k})>\frac{\psi}{(k+1)^{v}}\label{eq: inequality weak active with activity =00005Cpsiandrate=00005Cnu}
\end{equation}

\noindent for all histories $\tilde{\omega}^{k}$ such that $\tilde{L}_{k}(\tilde{\omega})>0.$
\end{defn}
In words, a supermartingale  has activity $\psi$ and rate $v$ if
the probability of a jump of size $\frac{\psi}{(k+1)^{v}}$ at time
$k$ exceeds $\frac{\psi}{(k+1)^{v}}$ for almost all histories. $\tilde{L}$
being a supermartingale, is weakly decreasing in expectations. The
assumption that it is weakly active asserts that  it cannot decay
or grow too slowly. Namely, there exists a bound on the rate at which
$\tilde{L}_{k}$ converges in such a manner that $\tilde{L}_{k}$
must not go up or down relative to $\tilde{L}_{k-1}$ in a considerably
rate with a probability corresponding to that bounded rate. It should
be noted that the class of supermartingales corresponding to $v=0$
and a constant $\psi\in(0,1)$ includes those processes which were
defined in Definition \ref{Def: active super martingale}. In the
next theorem the uniformity result presented in Theorem \ref{Th:  A.1 (F=000026L)}
is extended to include the class of all weakly active supermartingles.

\vspace{0cm}

\begin{thm}
\textup{\label{Th: uniform rate for weak active supermartingale}Let
$\epsilon,\psi,v\in(0,1),$ and $l_{0}>0.$ Then, for all $\underset{-}{L}\in(0,l_{0})$
there is a time $K<\infty$ such that }

\textup{
\[
P(\{\omega:\ \underset{k>K}{sup}\tilde{L}_{k}(\omega)\leq\underset{-}{L}\})\geq1-\epsilon
\]
}

\noindent \textup{for every weakly active supermartingale $\tilde{L}$
with activity $\psi$ and rate $v$ with $\tilde{L}_{0}\equiv l_{0}.$
}\footnote{Note that Theorem \ref{Th: uniform rate for weak active supermartingale}
holds for every initial value $l_{0},$ in particular, here it is
implemented for $l_{0}\equiv1.$}
\end{thm}
The power of the theorem stems from the fact that the integer $K$,
depends only on the parameters $l_{0},\ \epsilon>0,\ \psi,\ v$ and
$\underset{-}{L}$, and is otherwise independent of the particular
supermartingale selected. The theorem asserts that if $\text{\ensuremath{\tilde{L}} }$
is a weakly active supermartingale, then there is a fixed time $K$
by which, with high probability, $\text{\ensuremath{\tilde{L}} }$
drops below $\underset{-}{L}$ and remains below $\underset{-}{L}$
for all future periods.

\vspace{0cm}

The next theorem provide a novel linkage between each bi-parametric
class of $(\psi,\nu)$ - informative private beliefs and the class
of weakly active supermartingaless with activity $\psi$ and rate
$\nu.$ In particular, the identification asserts that each induced
process derived from a pair of $(\psi,\nu)$ - informative private
belief is weakly active supermartingale with activity $\psi$ and
rate $\nu.$
\begin{thm}
\textup{\label{th: l_k is weakly active supermartingale}} Let $\psi,\nu\in(0,1).$
If ($F^{L},F^{H})$ is {\em$(\psi,\nu)$ - informative} then the
induced  process $l_{k}$ is weakly active supermartingale with activity
$\frac{\psi}{2}$ and rate $\nu.$
\end{thm}
\vspace{0cm}

It should be emphasized that the values of $\nu$ are restricted to
be varied between in $(0,1)$. If $\nu=0$ then inequality \eqref{eq: inequality weak active with activity =00005Cpsiandrate=00005Cnu}
takes the form of an active supermartingale which rules out equations
\eqref{eq:transitions converfes to 1}, if $\nu>1$ on the other hand,
then the process $l_{k}$ fails to satisfy the uniformity property
as stated in Theorem \ref{eq:transitions converfes to 1}.

\section{Efficiency}

The expected time until individuals stop taking the wrong action can
be either finite or infinite as the induced stochastic process converges
\textquoteleft slowly\textquoteright{} to the cascade sets. Conceptually,
this section further studies how the learning efficiency might be
evaluated. To this end, let us consider the following stopping times;
the first time of the correct action $\tau=inf\{t:a_{t}=M\},$ the
aggregate amount of the incorrect actions $N=|\{t:a_{t}\neq M\}|,$
the time to learn $T_{H}=min\{t:a_{n}=M$ for all $n\geq t\},$ and
the time of the first mistake $T_{1}=min\{t:a_{t}\neq M\}.$ The learning
efficiency is referred to the finite expected value of these random
variables.

The next theorem asserts that whenever the corresponding cdf's  are
$(\psi,v)$ - informative then the expected time of the first correct
guess is finite. And on top of that, under a mild condition on the
cdfs' derivative (which plainly do not admit fluctuations around the
cascade set), the expected time of the total number of the incorrect
actions is finite as well. Moreover, both finite expectations are
uniformly bounded by a bound which merely depends on the level of
the informativeness' parameters $(\psi,v)$.

For the finite expectation of $T_{H}$ the subsequent theorem asserts
that there exist (large enough) $\hat{\psi}$ and (small enough) $\hat{\nu}$
such that for any $\psi>\hat{\psi}$ and $\nu<\hat{\nu,}$ and for
any corresponding cdf's which are $(\psi,v)$ - informative, the expected
time to learn is finite.

Before proceeding to the first theorem some preliminary is required.
Given a prior and a pair $(F^{L},F^{H}),$ $F$ is denoted to be the
(unconditional) distribution of the private belief being $H$. In
addition, $F$ is called smoothly monotone near zero if $pF'(p)\rightarrow0$
as $p\rightarrow0.$

\vspace{0cm}

\begin{thm}
\textup{\label{Th: finite expected time of the fist correct action}
For all $\psi,\nu\in(0,1)$ there exists $K=K(\psi,\nu)$ such that
for all pairs $(F^{L},F^{H})$ which are $(\psi,v)$ - informative:}

\textup{a. $E^{H}[\tau]<K.$}

\textup{b. If in addition $F$ is smoothly monotone near zero then
$E^{H}[N]<K.$}
\end{thm}
\vspace{0cm}

The aforementioned efficiency has been first studied by \citet{Dinah-Rosenberg-2017}
who propose the finite integral of $\int\frac{1}{F}$ as a sufficient
and necessary condition under which the expected time of $\tau$ holds.
It turns out that this condition fails to hold in many prominent cases
of interest as in the leading example which appears in \citet{Lones-Smith-2000}
where $F(p)=p$. Theorem \ref{Th: finite expected time of the fist correct action}
asserts that the class of $(\psi,\nu)$ - informative private beliefs
meets this condition in the following stronger form: the bound $K$
is uniform in the sense that it is independent of the particular private
belief chosen.
\begin{prop}
\textup{\label{Prop: the existence of K} There exist $\hat{\psi},\hat{v}\in(0,1),$
and $K=K(\hat{\psi},\hat{v})>0,$ and $\alpha=\alpha(\psi,v)>2,$
such that for all $\psi>\hat{\psi}$ and $\nu<\hat{v,}$ and for all
pairs $(F^{L},F^{H})$ which are $(\psi,v)$ - informative,}

\textup{
\[
P^{H}(a_{t}\neq M)<K\frac{1}{t^{\alpha}}
\]
}

\noindent \textup{for all $t\geq1.$}
\end{prop}
\vspace{0cm}

\begin{thm}
\textup{\label{TH 8:  E=00005BT_L=00005D < K}There exist $\hat{\psi},\hat{v}\in(0,1)$
and $K=K(\hat{\psi},\hat{v})>0$ such that for all $\psi>\hat{\psi}$
and $\nu<\hat{v,}$ and for all pairs $(F^{L},F^{H})$ which are $(\psi,v)$
- informative, $E^{H}[T_{H}]<K.$ Furthermore, $K(\psi,v)$ decreases
as $\psi\rightarrow1$ and $\nu\rightarrow0.$}
\end{thm}
\vspace{0cm}

\begin{proof}[\textbf{Proof of Theorem \ref{TH 8:  E=00005BT_L=00005D < K}}]
 By proposition \ref{Prop: the existence of K} there exist $K=K(\psi,v)$,
and a pair $(F^{L},F^{H})$ which is $(\psi,v)$ - informative such
that for all $t\geq2,\ P^{H}(a_{t-1}\neq M)<K\cdot\frac{1}{(t-1)^{\alpha}}.$
Thus, denoting $K_{1}=1+K\stackrel[t=2]{\infty}{\sum}\frac{t}{(t-1)^{\alpha}}$
and since $\{T_{H}=t\}\subseteq\{a_{t-1}\neq M\}$ for all $t\geq2$
we obtain

\[
\begin{array}{ll}
E[T_{H}|s=H] & =\stackrel[t=0]{\infty}{\sum}t\cdot P^{H}(\{T_{H}=t\})\\
 & \leq P^{H}(\{T_{H}=1\})+\stackrel[t=2]{\infty}{\sum}t\cdot P^{H}(\{a_{t-1}\neq M\})\\
 & \leq1+K\stackrel[t=2]{\infty}{\sum}\frac{t}{(t-1)^{\alpha}}\\
 & \leq K_{1}.
\end{array}
\]

\noindent By a symmetric argument the same conclusion holds for $s=L.$
Consequently, the expected time to learn $E[T_{H}]$ is uniformly
bounded by a constant which solely depends on $(\psi,v)$ and the
result follows.
\end{proof}
\vspace{0cm}

As $T_{1}<T_{H}$ the following corollary is inferred.
\begin{cor}
\textup{There exist $\hat{\psi},\hat{v}\in(0,1)$ and $K=K(\hat{\psi},\hat{v})>0$
such that for all $\psi>\hat{\psi}$ and $\nu<\hat{v,}$ and for all
pairs $(F^{L},F^{H})$ which are $(\psi,v)$ - informative, $E^{H}[T_{1}]<K.$}
\end{cor}
\vspace{0cm}

\section{Concluding remarks}

The classical setting of asymptotic learning is revisited where individuals
eventually take the correct action and their belief converge to the
truth, regardless of their private information assigned to the incorrect
state of the world. Nevertheless, recent papers have shown that the
prospective time to learn may take ages to infinity. The paper proposes
a simple (bi-parametric) criterion on the private information structures
and focuses on the result in time to learn appraisal. Whenever the
private information is unbounded, the criterion constitutes a characterization
in which the induced time to learn shares a common sharp bound. It
is further argued that the learning is efficient. For a general information
structure, it provides a universal constant $K$ such that with any
desired degree of precision at any fixed time $n$, only two scenarios
hold; either the learning has begun or the proposed criterion fails
to hold in more than $K$ periods by $n$.

The underlying technical results are mathematical; A new class of
supermartingales (called weakly active) is introduced, linked, and
shown to admit a uniform rate of convergence. This extends an earlier
result of \citet{Fudenberg-Levine-1992}.

\vspace{0cm}

\bibliographystyle{abbrvnat}
\bibliography{On_uniform_boundedness_of_sequential_social_learning}

\vspace{0cm}

\section*{ \center{APPENDIX}
}

\appendix
\vspace{0cm}

\section{Selected proofs\label{Appendix A sec:Selected-proofs}}

Significant savings in computational cost, the proof of the Theorem
\ref{Th: existance of inforamtive-CDF} is provided for the case in
which $M=2,$ where the payoff function is given by \begin{center}
\begin{tabular}{ | c | c | c | }
\hline 
 & 1 & 2 \\  
\hline 
H & 0 & 1 \\  
\hline 
L & 1 & 0 \\ 
\hline
\end{tabular} 
\end{center} Thus, the utility of the action $a_{k}=1,2$ is $1$ if and only
if the state of the world is $H,L,$ respectively, and zero otherwise.
\begin{proof}[\textbf{\small{}Proof of Theorem \ref{Th: existance of inforamtive-CDF}}]
Fix $(\psi,\nu)\in(0,1)^{2}$ and for $\psi<b<a<1$ let

\[
f_{1}(k)=\begin{cases}
\begin{array}{l}
a,\\
b,\\
\frac{\psi}{k^{\nu}}
\end{array} & \begin{array}{l}
k=-1\\
k=0\\
k\geq1.
\end{array}\end{cases}
\]

\noindent Let $\{b_{k}\}_{k\geq1}$ be a positive increasing to infinity
sequence and consider the following function $f\colon\mathbb{\mathbb{Z}\rightarrow\mathbb{R_{>}}}$
\[
f(k)=\begin{cases}
\begin{array}{l}
f_{1}(k-1)-f_{1}(k),\\
(f_{1}(|k|-1)-f_{1}(|k|))\frac{1}{e^{b_{|k|}}},
\end{array} & \begin{array}{l}
k\geq0\\
k\leq-1.
\end{array}\end{cases}
\]

\noindent Observe that

\begin{equation}
\begin{array}{cl}
C:= & \stackrel[k=-\infty]{+\infty}{\sum}f(k)=\stackrel[k=-\infty]{-1}{\sum}[f_{1}(|k|-1)-f_{1}(|k|)]\frac{1}{e^{b_{|k|}}}+\stackrel[k=0]{\infty}{\sum}[f_{1}(k-1)-f_{1}(k)]\\
 & =\stackrel[k=1]{\infty}{\sum}[f_{1}(k-1)-f_{1}(k)]\frac{1}{e^{b_{k}}}+a\\
 & <\stackrel[k=1]{\infty}{\sum}[f_{1}(k-1)-f_{1}(k)]+a=b+a.
\end{array}\label{eq: C< a+b}
\end{equation}

\noindent Consider the following probability functions over the integers:
$P^{H}(k)=\frac{f(k)}{C},\ P^{l}(k)=\frac{f(-k)}{C},\ k\in\mathbb{Z},$
one of which is associated with an accumulated distribution function
$F^{H},\ F^{L},$ respectively. Note that for all $k\geq1$

\begin{equation}
\begin{array}{ll}
1-F^{H}(-k) & =\frac{1}{C}\stackrel[i=-k+1]{-1}{\sum}[f_{1}(|i|-1)-f_{1}(|i|)]\frac{1}{e^{b_{|i|}}}+\frac{1}{C}\stackrel[i=0]{\infty}{\sum}[f_{1}(i-1)-f_{1}(i)]\\
 & =\frac{1}{C}\stackrel[i=-k+1]{-1}{\sum}[f_{1}(|i|-1)-f_{1}(|i|)]\frac{1}{e^{b_{|i|}}}+\frac{f_{1}(-1)}{C}\\
 & =\frac{1}{C}\stackrel[i=1]{k-1}{\sum}[f_{1}(i-1)-f_{1}(i)]\frac{1}{e^{b_{i}}}+\frac{f_{1}(-1)}{C},
\end{array}\label{eq: 1-F^H(-t)}
\end{equation}

\noindent where $1-F^{H}(0)=\frac{1}{C}\stackrel[i=1]{\infty}{\sum}[f_{1}(i-1)-f_{1}(i)]=\frac{f_{1}(0)}{C}.$
As well as

\begin{equation}
\begin{array}{cl}
1-F^{L}(-k) & =\frac{1}{C}\stackrel[i=-k+1]{0}{\sum}[f_{1}(|i|-1)-f_{1}(|i|)]+\frac{1}{C}\stackrel[i=1]{\infty}{\sum}[f_{1}(i-1)-f_{1}(i)]\frac{1}{e^{b_{i}}}\\
 & =\frac{1}{C}\stackrel[i=0]{k-1}{\sum}[f_{1}(i-1)-f_{1}(i)]+\frac{1}{C}\stackrel[i=1]{\infty}{\sum}[f_{1}(i-1)-f_{1}(i)]\frac{1}{e^{b_{i}}}\\
 & =\frac{1}{C}(f_{1}(-1)-f_{1}(k-1))+\frac{1}{C}\stackrel[i=1]{\infty}{\sum}[f_{1}(i-1)-f_{1}(i)]\frac{1}{e^{b_{i}}}\\
 & =\frac{1}{C}(f_{1}(-1)-f_{1}(k-1))+\frac{1}{C}(C-f_{1}(-1))=\frac{C-f_{1}(k-1)}{C}
\end{array}\label{eq: 1-F^L(-t)}
\end{equation}

\noindent with $1-F^{L}(0)=\frac{C-f_{1}(-1)}{C}.$ From \eqref{eq: 1-F^H(-t)}
and \eqref{eq: 1-F^L(-t)} we obtain

\begin{equation}
|1-F^{H}(-k)-1-F^{H}(-k)|=\frac{1}{C}[f_{1}(k-1)-\stackrel[i=k]{\infty}{\sum}[f_{1}(i-1)-f_{1}(i)]\frac{1}{e^{b_{i}}}]\label{eq: 1-F^=00007BH=00007D(-t) - 1-F^=00007BH=00007D(-t)}
\end{equation}

Now, for all $s\in\{L,H\}$ let

\[
\tilde{F}^{s}(p)=\begin{cases}
\begin{array}{l}
0,\\
1,\\
F^{s}(log(\frac{p}{1-p})),
\end{array} & \begin{array}{l}
0\\
1\\
other
\end{array}\end{cases}
\]

\noindent be the corresponding cdf on $[0,1],$ and observe that,
due to \eqref{eq:threreshold in P}, a simple calculation shows that
the corresponding threshold is given by $\bar{p}_{1}(l)=\frac{l}{1+l}.$
Hence, since

\[
\begin{array}{l}
|\rho(2|l_{k},H)-\rho(2|l_{k},L)|=|(\tilde{F}^{H}(p_{2}(l_{k}))-\tilde{F}^{H}(\bar{p}_{1}(l_{k})))-(\tilde{F}^{L}(p_{2}(l_{k}))-\tilde{F}^{L}(\bar{p}_{1}(l_{k})))|\\
\\
=|(1-\tilde{F}^{H}(\bar{p}_{1}(l_{k})))-(1-\tilde{F}^{L}(\bar{p}_{1}(l_{k})))|=|(1-\tilde{F}^{H}(\frac{l_{k}}{1+l_{k}}))-(1-\tilde{F}^{L}(\frac{l_{k}}{1+l_{k}}))|\\
\\
=|(1-F^{H}(log(\frac{(\frac{l_{k}}{1+l_{k}}}{1-(\frac{l_{k}}{1+l_{k}}})))-(1-F^{L}(log(\frac{(\frac{l_{k}}{1+l_{k}}}{1-(\frac{l_{k}}{1+l_{k}}})))|=|(1-F^{H}(log(l_{k})))-(1-F^{L}(log(l_{k})))|
\end{array}
\]

\noindent then, following Definition \ref{DEF: informative CDFs},
we need to show that

\[
|(1-F^{H}(log(l_{k})))-(1-F^{H}(log(l_{k})))|>\frac{\psi}{(k+1)^{\nu}},
\]

\noindent which from \eqref{eq: 1-F^=00007BH=00007D(-t) - 1-F^=00007BH=00007D(-t)}
and the fact that

\[
\begin{array}{l}
|(1-F^{H}(log(l_{k})))-(1-F^{H}(log(l_{k})))|=|(1-F^{H}(-(-log(l_{k})))-(1-F^{H}(-(-log(l_{k})))|\\
=|(1-F^{H}(-(log(\frac{1}{l_{k}})))-(1-F^{H}(-(log(\frac{1}{l_{k}})))|
\end{array}
\]

\noindent becomes

\begin{equation}
f_{1}(log(\frac{1}{l_{k}})-1)-\stackrel[i=log(\frac{1}{l_{k}})]{\infty}{\sum}[f_{1}(i-1)-f_{1}(i)]\frac{1}{e^{b_{i}}}|>\frac{\psi}{(k+1)^{\nu}}C.\label{eq: f1(log1/t - 1) - sigma >=00005Cpsi/(t+1)^vC}
\end{equation}
We will first show that $f_{1}(log(\frac{1}{l_{k}})-1)>\frac{\psi}{(k+1)^{\nu}}C$
for all $k\geq1,$ then we are provided by a sufficiently large sequence
$\{b_{k}\}_{k\geq1}$ such that inequality \eqref{eq: f1(log1/t - 1) - sigma >=00005Cpsi/(t+1)^vC}
holds for all $k\geq1.$ Note that for all $k\geq1$

\[
\begin{array}{l}
f_{1}(log(\frac{1}{l_{k}})-1)>\frac{\psi}{(k+1)^{\nu}}C\iff\frac{\psi}{(log(\frac{1}{l_{k}})+1)^{\nu}}>\frac{\psi}{(k+1)^{\nu}}C\\
\iff C(log(\frac{1}{l_{k}})-1)^{\nu}<(k+1)^{\nu}\\
\iff C^{\frac{1}{v}}[log(\frac{1}{l_{k}})-1]<k+1\\
\iff log(\frac{1}{l_{k}})<\frac{k+1}{C^{\frac{1}{v}}}+1\\
\iff\frac{1}{l_{k}}<e^{(\frac{k+1}{C^{\frac{1}{v}}}+1)}.
\end{array}
\]

Now, using equation \eqref{eq:iteratively of ln}, set (iteratively)
$\frac{1}{\tilde{l}_{k}}=\stackrel[k=i]{k}{\prod}\frac{1-\tilde{F}^{H}(\bar{p}_{1}(\tilde{l}_{i}))}{1-\tilde{F}^{L}(\bar{p}_{1}(\tilde{l}_{i}))}=\stackrel[k=i]{k}{\prod}\frac{1-F^{H}(log(\tilde{l}_{i}))}{1-F^{L}(log(\tilde{l}_{i}))}$
starting at $\tilde{l}_{0}=1$ (corresponding to $\bar{p}_{1}(l_{0})=\frac{l_{0}}{1+l_{0}}=\frac{1}{2}).$
Hence, since for all $k\geq0,$ $\tilde{l}_{k}\leq l_{k}$ and so
$\frac{1}{l_{k}}\leq\frac{1}{\tilde{l}_{k}},$ it is enough to show
that for all $k\geq1$

\[
\frac{1}{\tilde{l}_{k}}<e^{(\frac{k+1}{C^{\frac{1}{v}}}+1)}.
\]

Note that from \eqref{eq:iteratively of ln} we have

\[
log(\frac{1}{\tilde{l}_{k+1}})=log(\frac{1}{\tilde{l}_{k}}(\frac{1-F^{H}(log(\tilde{l}_{k})}{1-F^{L}(log(\tilde{l}_{k})}))=log(\frac{1}{\tilde{l}_{k}})+log(\frac{1-F^{H}(log(\tilde{l}_{k})}{1-F^{L}(log(\tilde{l}_{k})}).
\]

\noindent Hence, by judicious choice of $a,b,$ and $\{b_{k}\}_{k\geq1},$
it is sufficient to show that the sequence $\{log(\frac{1-F^{H}(log(\tilde{l}_{k})}{1-F^{L}(log(\tilde{l}_{k})}\}_{k\geq1}$
is monotone decreasing to zero and that $\frac{1}{\tilde{l}_{1}}=log(\frac{1-F^{H}(log(\tilde{l}_{0})}{1-F^{L}(log(\tilde{l}_{0})})<e^{(\frac{2}{C^{\frac{1}{v}}}+1)}.$
For the first part, it is sufficient to show that $\{log(\frac{1-F^{H}(-k)}{1-F^{L}(-k)})\}_{k\geq1}$
is monotone converges to zero. To this end, from \eqref{eq: 1-F^H(-t)}
and \eqref{eq: 1-F^L(-t)} and the fact that $\stackrel[i=1]{\infty}{\sum}[f_{1}(i-1)-f_{1}(i)]\frac{1}{e^{b_{i}}}=C-f_{1}(-1)$
we obtain

\[
\begin{array}{l}
\underset{k\rightarrow\infty}{lim}log(\frac{1-F^{H}(-k)}{1-F^{L}(-k)})=log(\frac{\stackrel[i=1]{k-1}{\sum}[f_{1}(i-1)-f_{1}(i)]\frac{1}{e^{b_{i}}}+f_{1}(-1)}{C-f_{1}(k-1)})=log(\frac{\stackrel[i=1]{\infty}{\sum}[f_{1}(i-1)-f_{1}(i)]\frac{1}{e^{b_{i}}}+f_{1}(-1)}{C})\\
=log(\frac{C}{C})=0.
\end{array}
\]

for the monotonic of $\{log(\frac{1-F^{H}(-k)}{1-F^{L}(-k)})\}_{k\geq1}$
we note that, since $f_{1}(k)$ decreases to zero it follows that
for all $k\geq1$ there exists a large enough $\bar{b}_{k}$ such
that for all $b_{k}>\bar{b}_{k}$ we have

\[
\begin{array}{l}
\frac{1-F^{H}(log(-k)}{1-F^{L}(log(-k)}=\frac{\stackrel[i=1]{k-1}{\sum}[f_{1}(i-1)-f_{1}(i)]\frac{1}{e^{b_{i}}}+f_{1}(-1)}{C-f_{1}(k-1)}\\
>\frac{\stackrel[i=1]{k-1}{\sum}[f_{1}(i-1)-f_{1}(i)]\frac{1}{e^{b_{i}}}+f_{1}(-1)+(f_{1}(k-1)-f_{1}(k))\frac{1}{e^{b_{k}}}}{(C-f_{1}(k-1))+(f_{1}(k-1)-f_{1}(k))}=\frac{\stackrel[i=1]{k}{\sum}[f_{1}(i-1)-f_{1}(i)]\frac{1}{e^{b_{i}}}+f_{1}(-1)}{C-f_{1}(k)}=\frac{1-F^{H}(log(-k-1)}{1-F^{L}(log(-k-1)}
\end{array}
\]

\noindent and since the $log$ is monotonic the result follows. For
the second part, observe that we need to show that

\[
f_{1}(log(\frac{1}{l_{1}})-1)-\stackrel[i=log(\frac{1}{l_{1}}))]{\infty}{\sum}[f_{1}(i-1)-f_{1}(i)]\frac{1}{e^{b_{i}}}|>\frac{\psi}{2^{v}}C
\]

\noindent hence we must show that $log(\frac{1}{\tilde{l}_{1}})>1$and
hence, by judicious choice of $a,b,\{b_{k}\}_{k\geq1},$ we will show
that $e<\frac{1}{\tilde{l}_{1}}<e^{(\frac{2}{C^{\frac{1}{v}}}+1)}.$
Now, observe that, setting $b_{1}=0,$ we have

\[
\begin{array}{l}
\frac{1}{\tilde{l}_{1}}=\frac{1-F^{H}(log(\tilde{l}_{0})}{1-F^{L}(log(\tilde{l}_{0})})=\frac{1-F^{H}(0)}{1-F^{L}(0)})=\frac{f_{1}(0)}{C-f_{1}(-1)}=\frac{f_{1}(0)}{\stackrel[i=1]{\infty}{\sum}[f_{1}(i-1)-f_{1}(i)]\frac{1}{e^{b_{i}}}}\\
=\frac{f_{!}(0)}{[f_{1}(0)-f_{1}(1)]\frac{1}{e^{b_{1}}}+\stackrel[i=2]{\frac{1}{\tilde{l}_{1}}}{\sum}[f_{1}(i-1)-f_{1}(i)]\frac{1}{e^{b_{i}}}+\stackrel[i=\frac{1}{\tilde{l}_{1}}+1]{\infty}{\sum}[f_{1}(i-1)-f_{1}(i)]\frac{1}{e^{b_{i}}}}=\frac{f_{1}(0)}{[f_{1}(0)-f_{1}(1)]+\stackrel[i=2]{\infty}{\sum}[f_{1}(i-1)-f_{1}(i)]\frac{1}{e^{b_{i}}}}
\end{array}.
\]

\noindent and so, setting $a=b+1$ and since $C<2a$ from \eqref{eq: C< a+b},
it follows that $e^{(\frac{2}{(2a)^{\frac{1}{v}}}+1)}=e^{(\frac{2}{(2b+2)}+1)}<e^{(\frac{2}{C^{\frac{1}{v}}}+1)}$
and hence it is sufficient to show that there exists $\psi<b,a$ such
that

\noindent 
\begin{equation}
e<\frac{f_{1}(0)}{f_{1}(0)-f_{1}(1)}=\frac{b}{b-\psi}<e^{(\frac{2}{(2b+2)}+1)}.\label{eq: e < b/(b-=00005Cpsi) < e^=00007B...=00007D}
\end{equation}

\noindent Indeed, by elementary consideration it can be shown that
\eqref{eq: e < b/(b-=00005Cpsi) < e^=00007B...=00007D} holds for
any $b$ in the interval $[\psi,\psi(\frac{e}{e-1})]$ which is sufficiently
close to $\psi(\frac{e}{e-1}).$ Denote such an element by $\tilde{b}$
and observe that there exists sufficiently large increasing sequence
$\{\tilde{b}_{k}\}_{k\geq2}$ such that for all $b_{k}>\tilde{b}_{k}$
one has

\[
e<\frac{1}{\tilde{l}_{1}}=\frac{b}{(b-\psi)+\stackrel[i=2]{\infty}{\sum}[f_{1}(i-1)-f_{1}(i)]\frac{1}{e^{b_{i}}}}<e^{(\frac{2}{(2b+2)}+1)}
\]

\noindent Now setting $b_{k}=min\{\tilde{b}_{k},\bar{b}_{k}\}$ obtain
that there exists $\hat{b}_{k}<b_{k}$ such that for all $k\geq1$

\[
f_{1}(log(\frac{1}{l_{k}})-1)-\stackrel[i=log(\frac{1}{l_{k}})]{\infty}{\sum}[f_{1}(i-1)-f_{1}(i)]\frac{1}{e^{\hat{b}_{i}}}|>\frac{\psi}{(k+1)^{\nu}}C
\]

\noindent and the result follows. It should be noted that whereas
the sequence $\{b_{k}\}_{k\geq1}$ is judiciously chosen to `enlarge'
the tail difference between the corresponding cdf's, the parameters
$a,b$ was chosen to bound the first value of $\frac{1}{l_{k}}.$
\end{proof}
\vspace{0.5cm}

\vspace{0cm}

\subsection{A uniform bound on the rate of convergence with respect to the class
of weakly active supermartingales with activity $\psi$ and rate $\nu$
with an initial value $l_{0}$}

\vspace{0cm}

Consider a general setting in which $\Omega=\{0,1\}$ where $(\Omega^{\infty},f,P)$
is equipped with a filtration $(f_{k})_{k\geq0}$ with $f=\sigma(\stackrel[k=0]{\infty}{\bigcup}f_{k}).$
\begin{defn}
A non-negative supermartingale $\{\tilde{L}_{k}\}_{k=0}^{\infty}$
is weakly active with activity $\psi$ and rate $\nu$ under $P$
if 

\[
P(\{\omega:|\frac{\tilde{L}_{k+1}(\omega)}{\tilde{L}_{k}(\omega)}-1|>\frac{\psi}{(k+1)^{\nu}}\}|\tilde{\omega}^{k})>\frac{\psi}{(k+1)^{\nu}}
\]

\noindent for almost all histories $\tilde{\omega}^{k}$ such that
$\tilde{L}_{k}(\tilde{\omega})>0$.
\end{defn}
In word, a supermartingale  has activity $\psi$ and rate $\nu$ if
the probability under $P$ of a jump of size $\frac{\psi}{(k+1)^{\nu}}$
at time $k$ exceeds $\frac{\psi}{(k+1)^{\nu}}$ for almost all histories. 

\vspace{0cm}

Before we state and prove the theorem about a uniform bound on the
rate of convergence, we will use some fundamental results from the
theory of supermartingales. In particular, bounds on the ``upcrossing
numbers'' which we introduce below. The result can be found in Neveu
(1975, Chapter 2). 
\begin{fact}
\textup{For any positive supermartingale $\widetilde{L}$ and any
$c>0,$ }

\textup{
\[
P(\underset{k}{sup}\widetilde{L}_{k}\geq c)\leq min\{1,\frac{\widetilde{L}_{0}}{c}\}.
\]
}
\end{fact}
\vspace{0cm}

For the next fact fix an interval $[a,b],\ 0<a<b<\infty,$ and define
the random variable: $U_{k}(a,b)(\omega)=$ the number of upcrossing
of $[a,b]$ of $\omega$ up to time $k;$ let $U_{\infty}(a,b)(\omega)=$
the total number of upcrossing of $[a,b]$ of $\omega$ (possibly
equal to $\infty).$ 
\begin{fact}
\textup{For any positive supermatingale $\widetilde{L}$ and $N>0,$ }

\textup{
\[
P(U_{\infty}(a,b)\geq N)\leq(\frac{a}{b})^{N}min\{1,\frac{\widetilde{L}_{0}}{a}\}.
\]
}

\noindent \textup{This is known as Dubin's inequality. }
\end{fact}
\vspace{0cm}

For the next lemma observe that a wekly active supermartingale $\widetilde{L}$
with activity $\psi$ and rate $\nu$ makes jump of size $\frac{\psi}{k^{\nu}}$
at time $k$ with probability at least $\frac{\psi}{k^{\nu}}$ in
each period $k$ where $\widetilde{L}_{k}>0.$ Consequently, over
a large number of periods either $\widetilde{L}$ has jumped to zero
(and stay there since $\widetilde{L}$ is a supermartingale) or there
are likely to be ``many'' jumps. Formally, given a time $k>0$,
define the random variable $J_{k}(\omega)$ to the number of times
$k'<k$ that $|\frac{\widetilde{L}_{k'}(\omega)}{\widetilde{L}_{k'-1}(\omega)}-1|\geq\frac{\psi}{k'^{\nu}},$
that is $J_{k}$ count the number of time by time $k$ that $\widetilde{L}$
faces jumps at rate of at least $\frac{\psi}{k'^{\nu}}.$ As a weakly
active supermartingale is a point wise definition, the next lemma
illustrates how the paths of any active supermartingale $\widetilde{L}(\omega)$
look like. The lemma asserts that there are only two types of paths
($P-a.s)$, for any given $J>0$ and $\epsilon\in(0,1),$ there exists
$K$ such that either that $\widetilde{L}_{K}(\omega)=0$, and since
$\widetilde{L}$ is a supermartingale there exists $k'<K$ such that
$\widetilde{L}_{k'}(\omega)=\widetilde{L}_{k'+1}(\omega)=...\widetilde{L}_{K}(\omega)=0,$
or $\widetilde{L}_{k}(\omega)>0$ for all $1\leq k\leq K$ and hence
there exist at least $J$ jumps, that is there exist a sub sequence
$(k_{i})_{i=1}^{J}\subseteq\{1,...,K\}$ such that $|\frac{\widetilde{L}_{k_{i}}(\omega)}{\widetilde{L}_{k_{i}-1}(\omega)}-1|\geq\frac{\psi}{k_{i}^{\nu}},\ \forall i\in\{1,...,J\}.$
The next lemma formalizes the above idea.

\vspace{0.5cm}

\begin{lem}
\textup{\label{(Lemma-A.4.-FL)} Let $\widetilde{L}$ be a weakly
active supermartingale with activity $\psi$ and rate $\nu.$ Then,
for all $0<\epsilon<1$ and $J>0$ there exists an integer $K$ such
that }

\textup{
\[
P(\{J_{K}\geq J\}\ or\ \{\widetilde{L}_{K}=0\})\geq1-\epsilon.
\]
}
\end{lem}
\begin{proof}
Let $0<\epsilon<1$ and $J>0.$ Because $\widetilde{L}$ has activity
$\psi$ and rate $\nu$ in each period $k,$ then either $\widetilde{L}_{k}=0$
or the probability of a jump of size $\frac{\psi}{k^{v}}$ at time
$k$ exceeds $\frac{\psi}{k^{v}}.$ Define a sequence of indicators
random variables $I_{k},\ k>0$ by 

\[
I_{k}(\omega)=\begin{cases}
\begin{array}{l}
1,\ \ \widetilde{L}_{k}(\omega)=0\ or\ |\frac{\widetilde{L}_{k}(\omega)}{\widetilde{L}_{k-1}(\omega)}-1|>\frac{\psi}{k^{v}}\\
0,\ \ otherwise.
\end{array}\end{cases}
\]

\noindent Now observe that since $\widetilde{L}$ has activity $\psi$
and rate $\nu$ then for each $k>0$,

\[
\begin{array}{l}
E^{P}[I_{k}]=1P(I_{k}=1)+0P(I_{k}=0)=P(\{\omega:I_{k}(\omega)=1\})=P(\{\omega:\widetilde{L}_{k}(\omega)=0\ or\ |\frac{\widetilde{L}_{k}(\omega)}{\widetilde{L}_{k-1}(\omega)}-1|>\frac{\psi}{k^{\nu}}\})=\\
P(\{\omega:\widetilde{L}_{k-1}(\omega)=0\}\cup\{\omega:\widetilde{L}_{k-1}(\omega)>0,\ and\ |\frac{\widetilde{L}_{k}(\omega)}{\widetilde{L}_{k-1}(\omega)}-1|>\frac{\psi}{k^{v}}]\})=\\
P(\{\omega:\widetilde{L}_{k-1}(\omega)=0\})+P(\{\omega:\widetilde{L}_{k-1}(\omega)>0,\ and\ |\frac{\widetilde{L}_{k}(\omega)}{\widetilde{L}_{k-1}(\omega)}-1|>\frac{\psi}{k^{v}}]\})\geq\\
\underset{\omega^{k-1}:\widetilde{L}_{k-1}(\omega)=0}{\sum}P(\omega^{k-1})+\underset{\omega^{k-1}:\widetilde{L}_{k-1}(\omega)>0}{\sum}P(\omega^{k-1})P(\{\bar{\omega:}|\frac{\widetilde{L}_{k}(\bar{\omega})}{\widetilde{L}_{k-1}(\bar{\omega})}-1|>\frac{\psi}{k^{v}}\}|\omega^{k-1})>\\
\underset{\omega^{k-1}:\widetilde{L}_{k-1}(\omega)=0}{\sum}P(\omega^{k-1})\frac{\psi}{k^{v}}+\underset{\omega^{k-1}:\widetilde{L}_{k-1}(\omega)>0}{\sum}P(\omega^{k-1})\frac{\psi}{k^{v}}=\frac{\psi}{k^{v}}\underset{\omega^{k-1}}{\sum}P(\omega^{k-1})=\frac{\psi}{k^{v}},
\end{array}
\]

\noindent thus, the expectation for each $I_{k}$ is more than $\frac{\psi}{k^{v}}$
and so

\noindent 
\begin{equation}
E^{p}[I_{1}]+...+E^{p}[I_{k}]>\frac{\psi}{1}+...+\frac{\psi}{k^{v}},\ \forall k>0.\label{eq:5}
\end{equation}

\noindent Now, using \eqref{eq:5} and applying the kolmogorov's strong
law, (Jiming 2010, Chapter 6, Theorem 6.7) for the sequences $(I_{k})_{k>0}$
and $\{a_{k}:=1+...+\frac{1}{k^{v}}\}_{k>0},$ we obtain

\[
\begin{array}{l}
1=P(\{\omega:\underset{k\rightarrow\infty}{lim}\frac{(I_{1}(\omega)-E^{p}[I_{1}])+...+(I_{k}(\omega)-E^{p}[I_{k}])}{a_{k}}=0\})=\\
P(\{\omega:\underset{k\rightarrow\infty}{lim}\frac{I_{1}(\omega)+...+I_{k}(\omega)}{a_{k}}=\underset{k\rightarrow\infty}{lim}\frac{E^{p}[I_{1}]+...+E^{p}[I_{k}]}{a_{k}}\geq\underset{k\rightarrow\infty}{lim}(\frac{\frac{\psi}{1}+...+\frac{\psi}{k^{v}}}{a_{k}})=\psi\underset{k\rightarrow\infty}{lim}(\frac{1+...+\frac{1}{k^{v}}}{a_{k}})=\psi\}),
\end{array}
\]

\noindent and so

\[
\begin{array}{l}
P(\{\omega:\underset{k\rightarrow\infty}{lim}\frac{I_{1}(\omega)+...+I_{k}(\omega)}{a_{k}}\geq\psi\})=1\iff\\
\\
P(\{\omega:\exists n_{\omega}>0\ s.t\ \forall k\geq n_{\omega},\ \frac{I_{1}(\omega)+...+I_{k}(\omega)}{a_{k}}\geq\psi\})=1\iff\\
\\
P(\{\omega:\exists n_{\omega}>0\ s.t\ \forall k\geq n_{\omega},\ I_{1}(\omega)+...+I_{k}(\omega)\geq\psi a_{k}\})=1.
\end{array}
\]

\noindent Now, denote $B_{k}:=\{\omega:\ I_{1}(\omega)+...+I_{k}(\omega)\geq\psi a_{k}>J\},\ k>0$
and observe that since $a_{k}\uparrow\infty$ and for any $\omega$
such that $\exists n_{\omega}>0\ s.t\ \boldsymbol{\forall k}\geq n_{\omega},\ I_{1}(\omega)+...+I_{k}(\omega)\geq\psi a_{k}$
it follows that for a fixed $J$ there exists large enough $n_{\omega}<n_{\omega}^{J}$
such that $\forall k\geq n_{\omega}^{J},\ I_{1}(\omega)+...+I_{k}(\omega)\geq J$
it follows that\footnote{Equivalently, $\{\omega:\exists n_{\omega}>0\ s.t\ \forall k\geq n_{\omega},\ I_{1}(\omega)+...+I_{k}(\omega)\geq\psi a_{k}\}\subset\{\omega:\exists n_{\omega}^{J}>0\ s.t\ \forall k\geq n_{\omega}^{J},\ I_{1}(\omega)+...+I_{k}(\omega)\geq J\}$}

\[
\begin{array}{l}
P(\{\omega:\exists n_{\omega}^{J}>0\ s.t\ \forall k\geq n_{\omega}^{J},\ I_{1}(\omega)+...+I_{k}(\omega)\geq J\})=1\iff\\
\\
P(\stackrel[n=1]{\infty}{\cup}\stackrel[k=n]{\infty}{\cap}B_{k})=\underset{n\rightarrow\infty}{lim}P(\stackrel[k=n]{\infty}{\cap}B_{k})=1.
\end{array}
\]

Hence, for $0<\epsilon<1$ and $J>0$ there exists sufficiently large
$K>0$ such that 

\[
1-\epsilon\leq P(\stackrel[k=K]{\infty}{\cap}B_{k})\leq P(B_{K})=P(\{\omega:\ I_{1}(\omega)+...+I_{K}(\omega)\geq J\})=P(\{\omega:S_{K}(\omega)\geq J\})=P(S_{K}\geq J),
\]

\noindent and therefore setting the random variable $S_{K}(\omega)=\underset{k<K}{\sum}I_{k}(\omega)$
we obtain

\begin{equation}
\begin{array}{l}
1-\epsilon\leq P(\{\omega:S_{K}(\omega)\geq J\})=P(\{\omega:\exists(k_{i})_{i=1}^{J}\subseteq\{1,...,K\}\ s.t\ I_{k_{i}}(\omega)=1,\ \forall1\leq i\leq J\})=\\
\\
P(\{\omega:\exists(k_{i})_{i=1}^{J}\subseteq\{1,...,K\}\ s.t.\ \boldsymbol{either}\ \exists1\leq i\leq J\ s.t\ 0=\widetilde{L}_{k_{i}}(\omega)=\widetilde{L}_{k_{i}+1}(\omega)=,...,=\widetilde{L}_{K}(\omega)\ \boldsymbol{or}\ \\
\\
\widetilde{L}_{k}(\omega)>0,\ \forall1\leq k\leq K\ and\ |\frac{\widetilde{L}_{k_{i}}(\omega)}{\widetilde{L}_{k_{i}-1}(\omega)}-1|>\frac{\psi}{k_{i}^{\nu}},\ \forall1\leq i\leq J\}).
\end{array}\label{eq:6}
\end{equation}

\noindent In other words, by equation \eqref{eq:6}, either $\widetilde{L}_{K}(\omega)=0$,
in which case, there exists $1\leq i\leq J$ such that $0=\widetilde{L}_{k_{i}}(\omega)=\widetilde{L}_{k_{i}+1}(\omega)=,...,=\widetilde{L}_{K}(\omega)$,
 or $\widetilde{L}_{k}(\omega)>0,\ \forall1\leq k\leq K$ and there
are $J$ jumps in all periods $(k_{i})_{i=1}^{J}$ (and so there exists
at least $J$ jumps in $1\leq k\leq K),$ that is, $|\frac{\widetilde{L}_{ki}(\omega)}{\widetilde{L}_{k_{i}-1}(\omega)}-1|>\frac{\psi}{k_{i}^{\nu}},\ \forall1\leq i\leq J.$
\end{proof}
\vspace{0cm}

\begin{lem}
\textup{\label{Lemma 5}Let $\epsilon,\psi\in(0,1),$ and for any
$0<\underline{c}<\overline{c}$ divide the interval $[\underline{c},\overline{c}]$
into $I$ equal sub-intervals with endpoints $e_{1}=\underline{c}<e_{2}<,...,<e_{I+1}=\overline{c}.$
Define the following events: }

\noindent \textup{$1.\ E_{\overline{c},K}^{1}=\{\omega:\underset{k\leq K}{max}\widetilde{L}_{k}(\omega)\geq\overline{c}\},$}

\noindent \textup{$2.\ E_{\underline{c},\overline{c}N,K,I}^{2}=\{\omega:\ \exists i\in\{1,..,I+1\}\ s.t\ [e_{i},e_{i+1}]\ is\ upcrossed\ by\ \widetilde{L}(\omega)\ N\ or\ more\ times\ by\ time\ K\},$}

\noindent \textup{$3.\ E_{J,K}^{3}=\{\omega:J_{K}(\omega)<J\ and\ \widetilde{L}_{K}(\omega)>0\},$}

\noindent \textup{$4.\ E_{\underline{c},K}^{4}=\{\omega:\underset{k\leq K}{min}\widetilde{L}_{k}(\omega)<\underline{c}\}.$}

\noindent \textup{Then, there exist judicious choice of $\underline{c},\overline{c},I,K,N,J$
such that}

\noindent \textup{(a) $(E_{\underline{c},K}^{4})^{c}\subset(E_{\overline{c},K}^{1}\cup E_{\underline{c},\overline{c}N,K,I}^{2}\cup E_{J,K}^{3}$).}

\noindent \textup{(b) $P(E_{\overline{c},K}^{1}),\ P(E_{\underline{c},\overline{c}N,K,I}^{2}),\ P(E_{J,K}^{3})\leq\frac{\epsilon}{4}$. }
\end{lem}
\vspace{0.3cm}

Using Lemma \ref{Lemma 5} and the results on up-crossing theorem,
which are reflected in facts 1 and fact 2, we will assert how the
paths of an weakly active supermatingale looks like. By judicious
choice of $\underline{c},\overline{c},I,K,N,J$ we will ensure that
$(E_{\underline{c},K}^{4})^{c}\subset(E_{\overline{c},K}^{1}\cup E_{\underline{c},\overline{c}N,K,I}^{2}\cup E_{J,K}^{3}$)
and that $P(E_{\overline{c},K}^{1}),P(E_{\underline{c},\overline{c}N,K,I}^{2}),P(E_{J,K}^{3})\leq\frac{\epsilon}{4}$.
We will show how this ensures that most paths of any weakly active
supermartingale are:

\noindent 1. Do not exceeds $\overline{c}$ for $\overline{c}$ large.
(Fact 1)

\noindent 2. Make ``few'' up-crossing of any positive interval $[a,b],$
(Fact 2) and

\noindent 3. either make ``lots of jumps'' or hit zero (Lemma 4). 

\noindent We will use these three conditions to show that there exists
large $K$ such that, most paths remains below$\underset{\sim}{L}$
from $K$ on.

\vspace{0.5cm}

\begin{proof}[\textbf{Proof of Part (a)}]
We will show that for any $\underline{c},\overline{c},K,J$ there
exist $I,N$ such that (a) holds. Let $\underline{c},\overline{c},K,J$
(these determines the sets $E_{\underline{c},K}^{4},E_{\overline{c},K}^{1},E_{J,K}^{3}$)
and suppose on the contrary that $(E_{\underline{c},K}^{4})^{c}\cap(E_{\overline{c},K}^{1})^{c}\cap(E_{\underline{c},\overline{c}N,K,I}^{2})^{c}\cap(E_{J,K}^{3})^{c}\neq\emptyset$
for all $I,N.$ Let $\omega\in(E_{\underline{c},K}^{4})^{c}\cap(E_{\overline{c},K}^{1})^{c}\cap(E_{J,K}^{3})^{c}$
and denote by $\widetilde{L}_{\omega,K}:=\{\widetilde{L}_{k}(\omega)\}_{0\leq k\leq K}$
to be the partial path of $\widetilde{L}(\omega)$ by time $K.$\footnote{Equivalently, let $\omega\in(E_{\underline{c},K}^{4})^{c}$ and assume
on the contrary that $\omega\notin(E_{\overline{c},K}^{1}\cup E_{\underline{c},\overline{c}N,K,I}^{2}\cup E_{J,K}^{3}),$
thus $\omega\in(E_{\overline{c},K}^{1}\cup E_{\underline{c},\overline{c}N,K,I}^{2}\cup E_{J,K}^{3})^{c}=(E_{\overline{c},K}^{1})^{c}\cap(E_{\underline{c},\overline{c}N,K,I}^{2})^{c}\cap(E_{J,K}^{3})^{c}$} Since $\omega\in(E_{\overline{c},K}^{1})^{c}\cap(E_{\underline{c},K}^{4})^{c}$
we have that $\underline{c}\leq\widetilde{L}_{\omega,K}\leq\overline{c}$.
In addition, since $\omega\in(E_{\overline{c},K}^{3})^{c}=\{\omega:J_{K}(\omega)>J\ or\ \widetilde{L}_{K}(\omega)=0\},$
and $0<\underline{c}\leq\widetilde{L}_{\omega,K}$, it follows that
the partial path $\widetilde{L}_{\omega,K}$ does not hit zero and
hence has $J$ or more jumps in $[\underline{c},\overline{c}]$ one
of which, corresponded to the k'th jump, is of relative size of at
least $\text{\ensuremath{\frac{\psi}{k}}}$ for all $k\in\{1,...,J\}.$ 

Now observe that in the range above $\underline{c}$ whenever $|\frac{\widetilde{L}_{K}(\omega)}{\widetilde{L}_{K-1}(\omega)}-1|>\frac{\psi}{K}$
we have

\begin{equation}
|\widetilde{L}_{K}(\omega)-\widetilde{L}_{K-1}(\omega)|>|\widetilde{L}_{K-1}(\omega)|\psi\geq\underline{c}\frac{\psi}{K}.\label{eq:L covers a subinterval}
\end{equation}

\noindent \footnote{Note that it does not necessarily mean that there was a jump in time
$K$ but a bound on the size of each sub-interval}In addition, recall that for any $I$ we have $I(e_{i}-e_{i-1})=\overline{c}-\underline{c}$,
thus if we choose

\noindent 
\begin{equation}
I\geq\frac{2\overline{c}}{\underline{c}\frac{\psi}{K}}+1\label{eq:I greater than}
\end{equation}

\noindent we then obtain that 

\begin{equation}
[\frac{2\overline{c}}{\underline{c}\frac{\psi}{K}}+1](e_{i}-e_{i-1})\leq I(e_{i}-e_{i-1})=\overline{c}-\underline{c}\label{eq:12}
\end{equation}

\noindent and so

\noindent 
\begin{equation}
(e_{i}-e_{i-1})\leq\frac{(\overline{c}-\underline{c})}{(\frac{2\overline{c}}{\underline{c}\frac{\psi}{K}}+1)}\leq\frac{(\overline{c}-\underline{c})}{(\frac{2\overline{c}}{\underline{c}\frac{\psi}{K}})}=\frac{\underline{c}\frac{\psi}{K}}{2}\frac{(\overline{c}-\underline{c})}{\overline{c}}\leq\frac{\underline{c}\frac{\psi}{K}}{2}.\label{eq:7}
\end{equation}

\noindent which yields that the width of each sub-interval is less
than $\frac{\underline{c}\text{\ensuremath{\frac{\psi}{K}}}}{2}.$

Therefore, each jump of $\widetilde{L}$ of a relative size of $\frac{\psi}{K}$
in the partial path $\widetilde{L}_{\omega,K}$ that remains between
$\underline{c}$ and $\overline{c}$ must cross (cover) at least one
of the sub-intervals  $[e_{i},e_{i+1}].$ Consequently, each $k'th$
jump of $\widetilde{L}$ of a relative size of $\frac{\psi}{k}$ in
the partial path $\widetilde{L}_{\omega,K}$ that remains between
$\underline{c}$ and $\overline{c}$ must cross (cover) at least $\frac{K}{k}$
sub-intervals  $[e_{i},e_{i+1}].$\footnote{A jump of a relative size of $\frac{\psi}{K}$ covers at least 1 interval.
Hence the first jump of a relative size of $\frac{\psi}{1}$ covers
at least $K$ intervals, the second jump of a relative size of $\frac{\psi}{2}$
covers at least $\frac{K}{2}$ intervals,...,the $k'th$ jump of a
relative size of $\frac{\psi}{k}$ covers at least $\frac{K}{k}$
intervals.} As a result since $\widetilde{L}_{\omega,K}$ has $J$ or more jumps
of a relative size of $\{\frac{\psi}{k}\}_{k=1}^{J}$ across sub-intervals
in $I$ by time $K$ it follows that $J$ jumps covers at least $\stackrel[k=1]{J}{\sum}\frac{K}{k}$
sub-intervals by time $K$. Therefore $\widetilde{L}_{\omega,K}$
must cross (in the worst case) at least one sub-interval $\lfloor\frac{\stackrel[k=1]{J}{\sum}\frac{K}{k}}{I}\rfloor\geq\frac{\stackrel[k=1]{J}{\sum}\frac{K}{k}}{I}-1=\frac{K\stackrel[k=1]{J}{\sum}\frac{1}{k}}{I}-1$
times, and hence, after $J$ jumps there must be at least one sub-interval
that is up-crossed, by $\widetilde{L}_{\omega,K}$ at least 

\[
\tilde{N}:=\frac{1}{2}(\frac{K\stackrel[k=1]{J}{\sum}\frac{1}{k}}{I}-1)-1
\]

\noindent times.\footnote{Notice that $\stackrel[k=1]{J}{\sum}\frac{K}{k}$ sub-intervals, denoted
$A,$ are dispersed among a total of $I$ sub-intervals, hence (in
the worst case) each sub-interval in $I$ is covered by at least $\frac{J}{I}$
elements from $A,$ and since $\frac{\stackrel[k=1]{J}{\sum}\frac{K}{k}}{I}$
might not be an integer we have that at least one interval in $I$
is covered $\lfloor\frac{\stackrel[k=1]{J}{\sum}\frac{K}{k}}{I}\rfloor\geq\frac{\stackrel[k=1]{J}{\sum}\frac{K}{k}}{I}-1$
times.} Consequently, we obtained that for $\omega\in(E_{\underline{c},K}^{4})^{c}\cap(E_{\overline{c},K}^{1})^{c}\cap(E_{J,K}^{3})^{c}$
we must have that $\omega\in E_{\underline{c},\overline{c},\tilde{N},K,I}^{2}$,
and so $\omega\notin(E_{\underline{c},\overline{c},\tilde{N},K,I}^{2})^{c}$.
Hence, if we choose any $N$ such that 

\noindent 
\begin{equation}
N<\tilde{N}\label{eq:N<}
\end{equation}

\noindent then, since $(E_{\underline{c},\overline{c},N,K,I}^{2})^{c}\subset(E_{\underline{c},\overline{c},\tilde{N},K,I}^{2})^{c}$
it follows that $\omega\notin(E_{\underline{c},\overline{c}N,K,I}^{2})^{c}$
and so $\omega\notin(E_{\underline{c},K}^{4})^{c}\cap(E_{\overline{c},K}^{1})^{c}\cap(E_{\underline{c},\overline{c}N,K,I}^{2})^{c}\cap(E_{J,K}^{3})^{c}$
which yields a contradiction. As a result, $(E_{\underline{c},K}^{4})^{c}\subset(E_{\overline{c},K}^{1}\cup E_{\underline{c},\overline{c}N,K,I}^{2}\cup E_{J,K}^{3})$
or equivalently $(E_{\overline{c},K}^{1})^{c}\cap(E_{\underline{c},\overline{c}N,K,I}^{2})^{c}\cap(E_{J,K}^{3})^{c}=(E_{\overline{c},K}^{1}\cup E_{\underline{c},\overline{c}N,K,I}^{2}\cup E_{J,K}^{3})^{c}\subset E_{\underline{c},K}^{4}$,
which is interpreted as follows: for any $\omega$ such that the partial
path $\widetilde{L}_{\omega,K}$ does not go above $\overline{c}$
,$(\omega\in(E_{\overline{c},K}^{1})^{c}),$ does not up-crossed any
interval more than $N$ times by time $K$, $(\omega\in(E_{\underline{c},\overline{c}N,K,I}^{2})^{c})$,
and jumps $J$ or more times by time $K$, $(\omega\in(E_{J,K}^{3})^{c}),$
must fall below $\underline{c}$, ($\omega\in E_{\underline{c},K}^{4}$).
\end{proof}
\vspace{0.5cm}

\vspace{0cm}

\begin{proof}[\textbf{Proof. }$P(E_{\overline{c},K}^{1})<\frac{\epsilon}{4}$]
 Applying Fact 1 with $\overline{c}_{1}=(\frac{4}{\epsilon})l_{0}$
and any $K_{1}>0$ we obtain that 

\[
\begin{array}{l}
P(E_{\overline{c}_{1},K_{1}}^{1})=P(\{\omega:\underset{0\leq k\leq K_{1}}{max}\widetilde{L}_{k}(\omega)\geq\overline{c}_{1}\})\leq P(\{\omega:\underset{k\geq0}{sup}\widetilde{L}_{k}(\omega)\geq\overline{c}_{1}\})\leq\\
min\{1,\frac{\widetilde{L}_{0}}{\overline{c}}\}=min\{1,\frac{l_{0}}{(\frac{4}{\epsilon})l_{0}}\}=min\{1,\frac{\epsilon}{4}\}=\frac{\epsilon}{4}.
\end{array}
\]

\noindent Note that the above inequalities hold regardless of how
we pick $K_{1}.$
\end{proof}
\vspace{0.5cm}

\begin{proof}[\textbf{Proof}. $P(E_{\underline{c},\overline{c}N,K,I}^{2})<\frac{\epsilon}{4}$]
\noindent  Following inequality \eqref{eq:I greater than} we set 

\begin{equation}
I=\frac{2\overline{c}}{\underline{c}\frac{\psi}{K}}+2\label{eq:15-1}
\end{equation}

\noindent and by diving the right hand side equality of inequality
\eqref{eq:12} by $Ie_{i}$ we obtain 

\[
0<\frac{e_{i-1}}{e_{i}}=1-\frac{(\overline{c}-\underline{c})}{Ie_{i}}<1-\frac{(\overline{c}-\underline{c})}{I\overline{c}}.
\]

\noindent Hence, applying Fact 2 with any interval $[e_{i+1},e_{i}]$
together with \eqref{eq:7} we obtain that the probability of $N$
or more up-crossings for any given sub-interval is not more than

\[
\begin{array}{l}
P(U_{\infty}(e_{i-1},e_{i})\geq N)=P(\{\omega:U_{\infty}(e_{i-1},e_{i})(\omega)\geq N\})\leq(\frac{e_{i-1}}{e_{i}})^{N}min\{1,\frac{\widetilde{L}_{0}}{e_{i-1}}\}<(\frac{e_{i-1}}{e_{i}})^{N}min\{1,\frac{\widetilde{L}_{0}}{\underline{c}}\}\leq\\
\\
(1-\frac{(\overline{c}-\underline{c})}{I\overline{c}})^{N}\frac{\widetilde{L}_{0}}{\underline{c}}.
\end{array}
\]

\noindent Consequently, the probability that some sub-interval is
up-crossed $N$ or more times is no more than 

\begin{equation}
\begin{array}{l}
P(\{\omega:At\ least\ one\ of\ the\ interval\ [e_{i-1},e_{i}]\ is\ upcrossed\ by\ \widetilde{L}(\omega)\ N\ or\ more\ times\ by\ time\ K\})=\\
\\
P(E_{\underline{c},\overline{c}N,K,I}^{2})\leq P(U_{\infty}(e_{i-1},e_{i})\geq N)\leq I(1-\frac{(\overline{c}-\underline{c})}{I\overline{c}})^{N}\frac{\widetilde{L}_{0}}{\underline{c}},
\end{array}\label{eq:regardless of K}
\end{equation}

\noindent and so in order to obtain $P(E_{\underline{c},\overline{c}N,K,I}^{2})\leq\frac{\epsilon}{4}$
we force $(1-\frac{(\overline{c}-\underline{c})}{I\overline{c}})^{N}\leq\frac{\epsilon}{4}\frac{1}{I}\frac{\underline{c}}{\widetilde{L}_{0}}$
where taking $log$ of both sides yields $N\geq\frac{log(\frac{\epsilon}{4}\frac{1}{I}\frac{\underline{c}}{\widetilde{L}_{0}})}{log(1-\frac{(\overline{c}-\underline{c})}{I\overline{c}})}.$
Hence taking any

\begin{equation}
N\geq\frac{log(\frac{\epsilon}{4}\frac{1}{I}\frac{\underline{c}}{\widetilde{L}_{0}})}{log(1-\frac{(\overline{c}-\underline{c})}{I\overline{c}})}\label{eq:14}
\end{equation}

\noindent yields that $P(E_{\underline{c},\overline{c}N,K,I}^{2})\leq\frac{\epsilon}{4}.$ 
\end{proof}
\vspace{0.3cm}

\begin{proof}[\textbf{Proof.} $P(E_{J,K}^{3})<\frac{\epsilon}{4}$]
 From inequalities \eqref{eq:N<} and \eqref{eq:14} we have

\begin{equation}
\frac{log(\frac{\epsilon}{4}\frac{1}{I}\frac{\underline{c}}{\widetilde{L}_{0}})}{log(1-\frac{(\overline{c}-\underline{c})}{I\overline{c}})}\leq N<\frac{1}{2}(\frac{K\stackrel[k=1]{J}{\sum}\frac{1}{k}}{I}-1)-1,\label{eq:N IS BOUNDED}
\end{equation}

\noindent thus set $N=\frac{log(\frac{\epsilon}{4}\frac{1}{I}\frac{\underline{c}}{\widetilde{L}_{0}})}{log(1-\frac{(\overline{c}-\underline{c})}{I\overline{c}})}$
and notice that $\frac{(2N+3)I}{K}\underset{K\rightarrow\infty}{\longrightarrow0,}$
hence there exist sufficiently large $J_{3}$ and sufficiently large
$K_{3}$ which satisfy the right hand side of \eqref{eq:N IS BOUNDED},
that is,$\frac{(2N+3)I}{K_{3}}<\stackrel[k=1]{J_{3}}{\sum}\frac{1}{k}$,
and so there exists $J_{4}>J_{3}>0$ such that $\frac{(2N+3)I}{K_{3}}<\stackrel[k=1]{J_{3}}{\sum}\frac{1}{k}<\stackrel[k=1]{J_{4}}{\sum}\frac{\psi}{k}$(check
for $\frac{\psi}{k^{v}}$ for $v>1)$.\footnote{Notice that from inequalities \ref{eq:N IS BOUNDED} the choice of
$J>2N(I+1)+2$ is arbitrary.} Now applying Lemma \eqref{(Lemma-A.4.-FL)} with $\frac{\epsilon}{4},\ J_{4}$
there exists a sufficiently large $K\geq K_{3}>0$ such that

\noindent 
\begin{equation}
P(\{\omega:J_{K}(\omega)\geq J_{3}\}\cup\{\omega:\widetilde{L}_{K}(\omega)=0\})=P(\{J_{K}\geq J_{3}\}\ or\ \{\widetilde{L}_{K}=0\})\geq1-\frac{\epsilon}{4},\label{eq:Large K}
\end{equation}
and so 

\[
\begin{array}{l}
P(E_{J_{3},K}^{3})=P((\{J_{K}\geq J_{3}\}\ or\ \{\widetilde{L}_{K}=0\})^{c})=P(\{J_{K}<J_{3}\}\ and\ \{\widetilde{L}_{K}>0\})=\\
\\
P(\{\omega:J_{K}(\omega)<J_{3}\ and\ \widetilde{L}_{K}(\omega)>0\})<\frac{\epsilon}{4}.
\end{array}
\]
\end{proof}
\vspace{0.3cm}

The next theorem poses a uniform bound on the convergence rate for
a class of weakly active supermartingales. This extends an earlier
Theorem \ref{Th:  A.1 (F=000026L)} result. The underlying technique
deployed in the proof illustrates how the paths of any weakly active
supermartingale look like as inferred from the corresponding point
wise definition.

\vspace{0.5cm}

\begin{proof}[\textbf{\small{}Proof of Theorem }{\small{}\ref{Th: uniform rate for weak active supermartingale} }]
 Given $\psi,l_{0}\underset{\sim}{L,\epsilon\in(0,1),}$ we conclude
that by choosing the following judicious choice of $\underline{c},\overline{c},I,K,N,J$
: $\underline{c}=\frac{\epsilon}{4}\underset{\sim}{L},$ and $\overline{c}=(\frac{4}{\epsilon})l_{0}$,
and $I=\frac{2\overline{c}}{\underline{c}\psi}+2$ (large enough),
and $N=\frac{I\frac{4}{\epsilon}\frac{\widetilde{L}_{0}}{\underline{c}}}{ln(1+\psi')}$
and $J=2I(N+1)+I+1$ (large enough), and large enough $K$ which satisfies
the conditions of \eqref{eq:Large K}, we obtain that 

\[
(E_{\underline{c},K}^{4})^{c}\subset E_{\overline{c},K}^{1}\cup E_{\underline{c},\overline{c}N,K,I}^{2}\cup E_{J,K\ }^{3}and\ P(E_{\overline{c},K}^{1}),P(E_{\underline{c},\overline{c}N,K,I}^{2}),P(E_{J,K}^{3})<\frac{\epsilon}{4}.
\]

\noindent This yields that

\begin{equation}
P(\{\omega:\underset{k\leq K}{min}\widetilde{L}_{k}(\omega)\geq\underline{c}\})=P((E_{\underline{c},K}^{4})^{c})<\frac{3\epsilon}{4}\label{eq:11}
\end{equation}

\noindent and so 

\begin{equation}
P(\{\omega:\underset{k\leq K}{min}\widetilde{L}_{k}(\omega)<\underline{c}\})=P(E_{\underline{c},K}^{4})\geq1-\frac{3\epsilon}{4}.\label{eq:10}
\end{equation}

\noindent Now, define the following stopping time $v(\omega)=min\{0\leq k\leq K:\widetilde{L}_{k}(\omega)<\underline{c}\}$
and observe that $\{\omega:\underset{k\leq K}{min}\widetilde{L}_{k}(\omega)<\underline{c}\}=\stackrel[k=1]{K}{\bigcup}\{v=k\}.$
In addition, note that for each $0\leq k'\leq K,$ $\{v=k'\}$ is
$g_{k}-adaptive$ and hence it can be decomposed in to finite disjoint
cylinders, denoted $\omega_{\{v=k'\}}^{j},$ and so $\widetilde{L}$
is a constant strictly less than $\underline{c}$ on each one of them.
Applying Fact 1 for $\underset{\sim}{L}$, $\underline{c}=\frac{\epsilon}{4}\underset{\sim}{L}$,
and the following supermartingale 

\[
\widetilde{L}_{k,k'}(\omega)=\widetilde{L}_{k+k'}(\omega)\boldsymbol{1}_{\omega_{\{v=k'\}}^{j}}(\omega)
\]

\noindent yields that 

\begin{equation}
P(\{\omega:\underset{k>K}{sup}\widetilde{L}_{k}(\omega)\leq\underset{\sim}{L}\}|\omega_{\{v=k'\}}^{j})\leq min(1,\frac{\widetilde{L}_{0,k'}}{\underset{\sim}{L}})<min(1,\frac{\underline{c}}{\underset{\sim}{L}})=min(1,\frac{\frac{\epsilon}{4}\underset{\sim}{L}}{\underset{\sim}{L}})=\frac{\epsilon}{4}.\label{eq:8}
\end{equation}

\noindent Since the sets $\{v=k\},\ k\geq0$ are disjoint we can infer
by taking the average of \eqref{eq:8} on all cylinders $\omega_{\{v=k'\}}^{j}$
corresponding to all $\{v=k\},$ that

\begin{equation}
P(\{\omega:\underset{k>K}{sup}\widetilde{L}_{k}(\omega)\leq\underset{\sim}{L}\}|\{\omega:\underset{k\leq K}{min}\widetilde{L}_{k}(\omega)<\underline{c}\})<\frac{\epsilon}{4},\label{eq:9}
\end{equation}

\noindent and we can conclude by \eqref{eq:9}, \eqref{eq:10}, \eqref{eq:11}
that 

\[
\begin{array}{l}
P(\{\omega:\underset{k>K}{max}\widetilde{L}_{k}(\omega)\leq\underset{\sim}{L}\})=P(\{\omega:\underset{k>K}{max}\widetilde{L}_{k}(\omega)\leq\underset{\sim}{L}\}|\{\omega:\underset{k\leq K}{min}\widetilde{L}_{k}(\omega)<\underline{c}\})P(\{\omega:\underset{k\leq K}{min}\widetilde{L}_{k}(\omega)<\underline{c}\})+\\
\\
P(\{\omega:\underset{k>K}{max}\widetilde{L}_{k}(\omega)\leq\underset{\sim}{L}\}|\{\omega:\underset{k\leq K}{min}\widetilde{L}_{k}(\omega)\geq\underline{c}\})P(\{\omega:\underset{k\leq K}{min}\widetilde{L}_{k}(\omega)\geq\underline{c}\})<\\
\\
\text{\ensuremath{\frac{\epsilon}{4}P(\{\omega:\underset{k\leq K}{min}\widetilde{L}_{k}(\omega)<\underline{c}\})+P(\{\omega:\underset{k>K}{max}\widetilde{L}_{k}(\omega)\leq\underset{\sim}{L}\}|\{\omega:\underset{k\leq K}{min}\widetilde{L}_{k}(\omega)\geq\underline{c}\})}\ensuremath{\frac{3\epsilon}{4}}}\leq\frac{\epsilon}{4}\cdot1+1\cdot\frac{3\epsilon}{4}=\epsilon.\\
\\
\end{array}
\]
\end{proof}
\vspace{0cm}

\subsection{Linkage between distance of measures and weakly active supermartingales}
\begin{defn}
For $\omega,k>0,P_{0},P_{1}\in\Delta(\Omega^{\infty})$ denote by

\[
\Delta(\omega^{k-1}):=||P_{1}(\cdot|\omega^{k-1})-P_{0}(\cdot|\omega^{k-1})||=\underset{a\in A}{max}\{|P_{1}(a|\omega^{k-1})-P_{0}(a|\omega^{k-1})|\}
\]

\noindent to be the distance between the conditional distributions
$P_{0},P_{1}$ over outcomes corresponding to $\omega^{k-1}.$
\end{defn}
\vspace{0cm}

The next lemma shows that given any history, $\tilde{\omega}^{k-1},$
the likelihood ratio is likely to substantially fall whenever $\Delta(\tilde{\omega}^{k-1})>\frac{\psi}{k^{\nu}}$
\begin{lem}
\textup{\label{Lemma: distance implies probability}Let $\psi,\nu\in(0,1),k>0,\tilde{\omega},P_{0},P_{1}$
such that $P_{0}(\tilde{\omega}^{k-1})>0$ and $\Delta(\tilde{\omega}^{k-1})>\frac{\psi}{k^{\nu}},$
then }

\textup{
\[
P_{0}(\{\omega:\ \frac{\tilde{L}_{k}(\omega)}{\tilde{L}_{k-1}(\omega)}-1\leq-\frac{\psi}{k^{\nu}\#A}\}|\tilde{\omega}^{k-1})\geq\frac{\psi}{k^{\nu}\#A}.
\]
}
\end{lem}
\begin{proof}
Denote $A=\{a_{1},a_{2},...,a_{M}\}$ and note that $\forall\omega,k>0$
with $\tilde{L}_{k-1}(\omega)>0$ one has $\frac{\tilde{L}_{k}(\omega)}{\tilde{L}_{k-1}(\omega)}=\frac{P_{1}(\omega_{k}|\omega^{k-1})}{P_{0}(\omega_{k}|\omega^{k-1})}$
and so conditional on $\tilde{\omega}^{k-1}$ the random variable
$\frac{\tilde{L}_{k}}{\tilde{L}_{k-1}}$ is distributed, with respect
to $P_{0},$ as follows:

\[
\frac{\tilde{L}_{k}}{\tilde{L}_{k-1}}|\tilde{\omega}^{k-1}\sim\begin{cases}
\begin{array}{l}
\frac{P_{1}(a_{1}|\tilde{\omega}^{k-1})}{P_{0}(a_{1}|\tilde{\omega}^{k-1})},\ \ \ P_{0}(a_{1}|\tilde{\omega}^{k-1})\\
\frac{P_{1}(a_{2}|\tilde{\omega}^{k-1})}{P_{0}(a_{2}|\tilde{\omega}^{k-1})},\ \ \ P_{0}(a_{2}|\tilde{\omega}^{k-1})\\
.\\
.\\
.\\
\frac{P_{1}(a_{M}|\tilde{\omega}^{k-1})}{P_{0}(a_{M}|\tilde{\omega}^{k-1})},\ \ \ P_{0}(a_{M}|\tilde{\omega}^{k-1}).
\end{array}\end{cases}
\]
Consequently, it suffices to show that for some $m\in\{1,...,M\}:$ 

\[
\frac{P_{1}(a_{m}|\tilde{\omega}^{k-1})}{P_{0}(a_{m}|\tilde{\omega}^{k-1})}\leq1-\frac{\psi}{k^{\nu}M}\ \ and\ \ P_{0}(a_{m}|\tilde{\omega}^{k-1})\geq\frac{\psi}{k^{\nu}M}.
\]

By hypothesis, $\Delta(\tilde{\omega}^{k-1})=\underset{a\in A}{max}\{|P_{1}(a|\tilde{\omega}^{k-1})-P_{0}(a|\tilde{\omega}^{k-1})|\}\geq\frac{\psi}{k^{\nu}}.$
\footnote{Lemma \ref{Lemma: distance implies probability} can be stated equivalently:
if there exists $m$ such that $|P_{1}(a_{m}|\tilde{\omega}^{k-1})-P_{0}(a_{m}|\tilde{\omega}^{k-1})|\geq\frac{\psi}{k^{\nu}}$
then $P_{0}(\{\omega:\ \frac{\tilde{L}_{k}(\omega)}{\tilde{L}_{k-1}(\omega)}-1\leq-\frac{\psi}{k^{\nu}\#A}\}|\tilde{\omega}^{k-1})\geq\frac{\psi}{k^{\nu}\#A}.$
Hence, the negation yields that for any $m:$ $|P_{1}(a_{m}|\tilde{\omega}^{k-1})-P_{0}(a_{m}|\tilde{\omega}^{k-1})|<\frac{\psi}{k^{\nu}}$.}Suppose without loss of generality that this maximum attains at $m=1,$
hence, $|P_{1}(a_{1}|\tilde{\omega}^{k-1})-P_{0}(a_{1}|\tilde{\omega}^{k-1})|\geq\frac{\psi}{k^{\nu}}.$ 

Case 1: If $P_{0}(a_{1}|\tilde{\omega}^{k-1})-P_{1}(a_{1}|\tilde{\omega}^{k-1})\geq\frac{\psi}{k^{\nu}}$
then

\[
\frac{\psi}{k^{\nu}M}\leq\frac{\psi}{k^{\nu}}+P_{1}(a_{1}|\tilde{\omega}^{k-1})\leq P_{0}(a_{1}|\tilde{\omega}^{k-1})
\]

\noindent where dividing the last inequality by $P_{0}(a_{1}|\tilde{\omega}^{k-1})$
yields that

\noindent 
\[
\frac{\psi}{k^{\nu}M}\leq\frac{\psi}{k^{\nu}}\leq\frac{\frac{\psi}{k^{\nu}}}{P_{0}(a_{1}|\tilde{\omega}^{k-1})}\leq1-\frac{P_{1}(a_{1}|\tilde{\omega}^{k-1})}{P_{0}(a_{1}|\tilde{\omega}^{k-1})}
\]

\noindent and we are done.

Case 2: If $P_{0}(a_{1}|\tilde{\omega}^{k-1})-P_{1}(a_{1}|\tilde{\omega}^{k-1})\leq-\frac{\psi}{k^{\nu}}$
then 

\[
\begin{array}{l}
0=1-1=\stackrel[m=1]{M}{\sum}P_{0}(a_{m}|\tilde{\omega}^{k-1})-\stackrel[m=1]{M}{\sum}P_{1}(a_{m}|\tilde{\omega}^{k-1})=\stackrel[m=1]{M}{\sum}(P_{0}(a_{m}|\tilde{\omega}^{k-1})-P_{1}(a_{m}|\tilde{\omega}^{k-1}))=\\
P_{0}(a_{1}|\tilde{\omega}^{k-1})-P_{1}(a_{1}|\tilde{\omega}^{k-1})+\stackrel[m=2]{M}{\sum}(P_{0}(a_{m}|\tilde{\omega}^{k-1})-P_{1}(a_{m}|\tilde{\omega}^{k-1}))\leq\\
-\frac{\psi}{k^{\nu}}+\stackrel[m=2]{M}{\sum}(P_{0}(a_{m}|\tilde{\omega}^{k-1})-P_{1}(a_{m}|\tilde{\omega}^{k-1})).
\end{array}
\]

Consequently,

\noindent 
\[
\begin{array}{l}
\frac{\psi}{k^{\nu}}\leq\stackrel[m=2]{M}{\sum}(P_{0}(a_{m}|\tilde{\omega}^{k-1})-P_{1}(a_{m}|\tilde{\omega}^{k-1}))\leq(M-1)\underset{2\leq m}{max}\{P_{0}(a_{m}|\tilde{\omega}^{k-1})-P_{1}(a_{m}|\tilde{\omega}^{k-1})\}\leq\\
M\underset{2\leq m}{max}\{P_{0}(a_{m}|\tilde{\omega}^{k-1})-P_{1}(a_{m}|\tilde{\omega}^{k-1})\},
\end{array}
\]
 and so $\frac{\psi}{k^{\nu}M}\leq\underset{2\leq m}{max}\{P_{0}(a_{m}|\tilde{\omega}^{k-1})-P_{1}(a_{m}|\tilde{\omega}^{k-1})\}.$
W.l.o.g assume that the maximum attains at $m=2.$ Hence, $\frac{\psi}{k^{\nu}M}\leq P_{0}(a_{2}|\tilde{\omega}^{k-1})-P_{1}(a_{2}|\tilde{\omega}^{k-1})$
which yields that

\[
\frac{\psi}{k^{\nu}M}\leq\frac{\psi}{k^{\nu}M}+P_{1}(a_{2}|\tilde{\omega}^{k-1})\leq P_{0}(a_{2}|\tilde{\omega}^{k-1}),
\]

\noindent where dividing the last inequality by $P_{0}(a_{2}|\tilde{\omega}^{k-1})$
we conclude once again that

\[
\frac{\psi}{k^{\nu}M}\leq\frac{\frac{\psi}{k^{\nu}}}{MP_{0}(a_{2}|\tilde{\omega}^{k-1})}\leq1-\frac{P_{1}(a_{2}|\tilde{\omega}^{k-1})}{P_{0}(a_{2}|\tilde{\omega}^{k-1})}.
\]
\end{proof}
\vspace{0cm}

Consider a general setting in which $\Omega=\{0,1\}$ where $(\Omega^{\infty},f,P_{0},P_{1})$
is equipped with a filtration $(f_{k})_{k\geq0}$ with $f=\sigma(\stackrel[k=0]{\infty}{\bigcup}f_{k}).$
\begin{lem}
\textup{\label{Lemma: L is supermartingale}The process $\tilde{L}_{0}\equiv1,\ \{\tilde{L}_{k}(\omega)$
$=\frac{P_{1}(\omega^{k})}{P_{0}(\omega^{k})}\}_{k>0}$, is a $P_{0}-$supermartingale. }
\end{lem}
\begin{proof}
It follows that for every $\omega,k>0:$

\[
\begin{array}{l}
E^{P_{0}}[\tilde{L}_{k}|\omega^{k-1}]=E^{P_{0}}[\tilde{L}_{k-1}\frac{P_{1}(\cdot|\omega^{k-1})}{P_{0}(\cdot|\omega^{k-1})}|\omega^{k-1}]=\tilde{L}_{k-1}\cdot E^{P_{0}}[\frac{P_{1}(\cdot|\omega^{k-1})}{P_{0}(\cdot|\omega^{k-1})}|\omega^{k-1}]=\\
\\
=\tilde{L}_{k-1}\underset{a\in\Omega:\,P_{0}(a|\omega^{k-1})>0}{\sum}P_{0}(a|\omega^{k-1})\frac{P_{1}(a|\omega^{k-1})}{P_{0}(a|\omega^{k-1})}=\tilde{L}_{k-1}\underset{\leq1}{\cdot\underbrace{\underset{a\in\Omega:\,P_{0}(a|\omega^{k-1})>0}{\sum}P_{1}(a|\omega^{k-1})}\leq\tilde{L}_{k-1}},
\end{array}
\]

\noindent and as a result $E^{P_{0}}[\tilde{L}_{k}|f_{k-1}]\leq\tilde{L}_{k-1},\ P_{0}-a.s.$
It should be noted that if $P_{1}\ll P_{0}$ then

\noindent $\underset{a\in\Omega:\,P_{0}(a|\omega^{k-1})>0}{\sum}P_{1}(a|\omega^{k-1})=1$
result in $\tilde{L}_{k}$ being a martingale.
\end{proof}
\vspace{0cm}

\begin{proof}[\textbf{\small{}Proof of Theorem \ref{th: l_k is weakly active supermartingale}}]

The fact that $l_{k}$ is a supermartingale follows directly from
equation \eqref{eq:iteratively of ln} and Lemma \ref{Lemma: L is supermartingale}.
To verify that $l_{k}$ is weakly active supermartingale with activity
$\frac{\psi}{2}$ and rate $\nu$ let $\tilde{\omega}^{k-1}$ such
that $l_{k-1}(\tilde{\omega})>0$ for some $k>0$. Observe that, since
the thresholds $\{\bar{p}_{1}(l)\}_{l\geq0}$ are strictly decreasing
then the monotonic of $\Delta(p)$ together with \eqref{eq: condition (psi,v) - informative}
yield

\begin{equation}
\Delta(\bar{p}_{1}(l_{k}))\geq\Delta(\bar{p}_{1}(\tilde{l}_{k}))>\frac{\psi}{(k+1)^{\nu}}\label{eq: |Deta l_t| >=00003D |Delta l^tilda_t| >=00003D psi/(t+1)^v}
\end{equation}

\noindent for all $k\geq0.$ Furthermore, since ($F^{L},F^{H})$ is
{\em$(\psi,\nu)$ - informative}, inequality \eqref{eq: |Deta l_t| >=00003D |Delta l^tilda_t| >=00003D psi/(t+1)^v}
yields

\[
\rho(2|l_{k-1}(\tilde{\omega}),H)-\rho(2|l_{k-1}(\tilde{\omega}),L)=\Delta(\bar{p}_{1}(l_{k-1}(\tilde{\omega})))\geq\Delta(\bar{p}_{1}(\tilde{l}_{k-1}))>\frac{\psi}{k^{\nu}}
\]

\noindent and hence

\noindent 
\[
\underset{m\in\{1,2\}}{max}\{|P^{H}(m|\tilde{\omega}^{k-1})-P^{L}(m|\tilde{\omega}^{k-1})|\}=\underset{m\in\{1,2\}}{max\{}|\rho(m|l_{k-1}(\tilde{\omega}),H)-\rho(m|l_{k-1}(\tilde{\omega}),L)|\}>\frac{\psi}{k^{\nu}}.
\]

Applying Lemma \ref{Lemma: distance implies probability} with $P^{H},P^{L},\tilde{\omega},$
and $\#A=2$ we obtain

\[
P^{H}(\{\omega:|\frac{l_{k}(\omega)}{l_{k-1}(\omega)}-1|>\frac{\psi}{k^{\nu}}\}|\tilde{\omega}^{k-1})\geq P^{H}(\{\omega:\frac{l_{k}(\omega)}{l_{k-1}(\omega)}-1\leq-\frac{\psi}{k^{\nu}}\}|\tilde{\omega}^{k-1})\geq\frac{\psi}{2k^{\nu}}
\]
 and the result follows.
\end{proof}
\vspace{0cm}
The proof of Theorem \ref{Th: uniform bound for (psi,nu) informative}
is directly followed. 
\begin{proof}[\textbf{\small{}Proof of Theorem }\ref{Th: uniform bound for (psi,nu) informative}\textbf{\small{} }]
\noindent Let $\psi,\nu\in(0,1).$ From Theorem \ref{th: l_k is weakly active supermartingale}
the induced process $l_{k}$ associated with every pair $(F^{L},F^{H})$
which is $(\psi,\nu)$ - informative is weakly active with the same
activity $\frac{\psi}{2}$ and rate $\nu$ started at $l_{0}=1.$
Hence, from the uniformity Theorem \ref{Th: uniform rate for weak active supermartingale}
for all $\bar{L}<1$ we are provided with a uniform finite time $K=K(\psi,\nu,\epsilon,\bar{L}),$
which solely depends on these variables, such that for all pairs $(F^{L},F^{H})$
which are $(\psi,\nu)$ - informative there is $P^{H}$ - probability
of at least $(1-\epsilon)$ that
\[
P(\underset{k>K}{sup}l_{k}\leq\bar{L})=P(\{\omega:\underset{k>K}{sup}l_{k}(\omega)\leq\bar{L}\})\geq1-\epsilon
\]

\noindent and the result follows.
\end{proof}
\vspace{0cm}

\section{Extracting weakly active supermartingales}

\vspace{0cm}

\subsection{\label{subsec:Construction}Construction }

Let $\psi,\nu\in(0,1)$ and define an increasing sequence of stopping
times $\{\tau_{k}\}_{k=0}^{\infty}$ relative to $\{L_{t}\}_{t\geq0}$
and $\{\epsilon_{t.}=\frac{\psi}{(t+1)^{\nu}}\}_{t\geq0}$ inductively
as follows: first set $\tau_{0}\equiv0$ and for all $k>0$ if $\tau_{k-1}(\omega)=\infty$
set $\tau_{k}(\omega)=\infty$ where for $\tau_{k-1}(\omega)<\infty$
define $\tau_{k}(\omega)$ to be the smallest integer $t>\tau_{k-1}(\omega)$
such that either 

\begin{equation}
P(\{\bar{\omega}:\,\frac{L_{t}(\bar{\omega})}{L_{t-1}(\bar{\omega})}-1<-\frac{\epsilon_{t}}{\#A}\}|\ \omega^{t-1})>\frac{\epsilon_{t}}{\#A}\label{condition1}
\end{equation}
or 

\noindent 
\begin{equation}
\frac{L_{t}(\omega)}{L_{\tau_{k-1}(\omega)}(\omega)}-1>\frac{\epsilon_{t}}{2\#A}.\label{condition 2-1}
\end{equation}
If there is no such $t,$ set $\tau_{k}(\omega)=\infty.$\footnote{Note that whenever condition \eqref{condition1} holds, $\tau_{k}(\hat{\omega})=t$
for all $\hat{\omega}\in\omega^{t-1}.$} 

Define the faster process: $\{\tilde{L}_{k}\}_{k=0}^{\infty}$ relative
to the process $\{L_{t}\}_{t\geq0}$ by

\[
\tilde{L}_{k}(\omega)=\begin{cases}
\begin{array}{l}
L_{\tau_{k}(\omega)}(\omega),\\
0,
\end{array} & \begin{array}{l}
\tau_{k}(\omega)<\infty\\
\tau_{k}(\omega)=\infty.
\end{array}\end{cases}
\]

\vspace{0cm}

Note that $\{\tilde{L}_{k}\}_{k=0}^{\infty}$ is defined with respect
to the stopping times $\{\tau_{k}\}_{k=0}^{\infty}$ which is adapted
to an associated filtration whose events are denoted by $\tilde{\omega}^{k}.$\footnote{$k$ is interpreted as the first time that $L$ satisfies one of the
three rules, where $t$ denotes a general time index.}

\vspace{0cm}

\begin{lem}
\textup{\label{Lemma: the faster process is weakly active supermartingale}For
all $t>0,$ the faster process $\{\tilde{L}_{k}\}_{k\geq0}$ relative
to the process $\{L_{t}\}_{t\geq0}$ and $\{\epsilon_{t}=\frac{\psi}{(t+1)^{\nu}}\}_{t\geq0}$
is weakly active supermartingale (under $P$) with activity $\frac{\psi}{2\#A}$
and rate $\nu.$}
\end{lem}
\begin{proof}
Since $\{\tau_{k}\}_{k=0}^{\infty}$ are stopping times and $\{L_{t}\}_{t\geq0}$
is a supermartingale, it follows that $\{\tilde{L}_{k}\}_{k=0}^{\infty}$
is a supermartingale. Now, let $\tilde{\omega}\in\Omega^{\infty},k>0$
such that $\tilde{L}_{k-1}(\tilde{\omega})>0$, we need to show that 

\begin{equation}
P(\{\omega:|\frac{\tilde{L}_{k}(\omega)}{\tilde{L}_{k-1}(\omega)}-1|>\frac{\epsilon_{k}}{2\#A}\}|\,\tilde{\omega}^{k-1})>\frac{\epsilon_{k}}{2\#A}.\label{eq:1}
\end{equation}

To see this, let $s=\tau_{k-1}(\tilde{\omega})$ and note that $\tau_{k-1}(\omega)=s,\ \forall\omega\in\tilde{\omega}^{s}.$
Since $\tilde{\omega}^{s}$ is cylinder it follows that $L_{s}(\omega)\equiv L_{s}(\tilde{\omega})>0,\ \forall\omega\in\tilde{\omega}^{s},$
and in addition, $\tau_{k}(\cdot|\tilde{\omega}^{s})$ is a random
variable which denotes the first time after $k-1$ conditional on
the event which consists all $\omega's$ for which $\tau_{k-1}(\omega)=s$
and $\omega^{s}\in\tilde{\omega}^{s}.$\footnote{$0<\tilde{L}_{k-1}(\tilde{\omega}):=L_{\tau_{k-1}(\tilde{\omega})}(\tilde{\omega}):=L_{s}(\tilde{\omega})$}
Hence, in terms of equation \eqref{eq:1} it is enough to show that 

\begin{equation}
P(\{\omega:|\frac{L_{\tau_{k}(\omega)}(\omega)}{L_{s}(\omega)}-1|>\frac{\epsilon_{t}}{2\#A}\}|\,\tilde{\omega}^{s})>\frac{\epsilon_{t}}{\#A}.\label{eq:2}
\end{equation}

\noindent where $L_{\infty}$ is defined to be the constant zero.
Now note that the sets:

\[
\begin{array}{l}
\{rule\ 1\}:=\{\omega:\tau_{k}(\omega)\ satisfies\ rule\ 1\},\\
\{rule\ 2\setminus1\}:=\{\omega:\tau_{k}(\omega)\ satisfies\ rule\ 2\ and\ does\ not\ satisfy\ rule\ 1\},\\
\{rule\ 3\}:=\{\omega:\tau_{k}(\omega)\ satisfies\ rule\ 3\},
\end{array}
\]

\noindent are disjoint, and that for any $\omega,$ one of these three
sets must be used to choose $\tau_{k}(\omega).$ Hence, it is enough
to show that equation \eqref{eq:2} holds conditional on each set
and thus it holds averaging over all of them. 

\noindent Conditional on $\{rule\ 3\}:$ 

\noindent 
\[
\begin{array}{l}
P(\{\omega:|\frac{L_{\tau_{k}(\omega)}(\omega)}{L_{s}(\omega)}-1|>\frac{\epsilon_{t}}{2\#A}\}|\,\tilde{\omega}^{s}\cap\{rule\ 3\})=\\
P(\{\omega:|\frac{0}{L_{s}(\omega)}-1|>\frac{\epsilon_{t}}{2\#A}\}|\,\tilde{\omega}^{s}\cap\{\omega:L_{\infty}(\omega)=0\})=1.
\end{array}
\]
 Conditional on $\{rule\ 2\setminus1\}:$ 

\[
\begin{array}{l}
1=P(\{\omega:\frac{L_{\tau_{k}(\omega)}(\omega)}{L_{s}(\omega)}-1>\frac{\epsilon_{t}}{2\#A}\}|\,\tilde{\omega}^{s}\cap\{rule\ 2\setminus1\})=\\
P(\{\omega:|\frac{L_{\tau_{k}(\omega)}(\omega)}{L_{s}(\omega)}-1|>\frac{\epsilon_{t}}{2\#A}\}|\,\tilde{\omega}^{s}\cap\{rule\ 2\setminus1\}).
\end{array}
\]

\noindent Conditional on $\{rule\ 1\}:$ since condition 1 holds we
obtain that $\tau_{k}(\hat{\omega})=t$ for all $\hat{\omega}\in\omega^{t-1}$
and thus

\begin{equation}
P(\{\omega:\frac{L_{\tau_{k}(\omega)}(\omega)}{L_{\tau_{k}(\omega)-1}(\omega)}-1<-\frac{\epsilon_{t}}{\#A}\}|\,\tilde{\omega}^{s}\cap\{rule\ 1\})\geq\frac{\epsilon_{t}}{\#A}.\label{eq:3}
\end{equation}

\noindent Now, note that if condition 2 was satisfied at time $\tau_{k}(\omega)-1$
then $\tau_{k}(\omega)$ would not be the infimum on all $t>\tau_{k-1}(\omega)$
such that condition 1 holds at time $t=\tau_{k}(\omega)$. Therefore,
condition 2 was not used at time $\tau_{k}(\omega)-1$ just before
time $\tau_{k}(\omega),$ which yields that

\[
\frac{L_{\tau_{k}(\omega)-1}(\omega)}{L_{s}(\omega)}-1\leq\frac{\epsilon_{s}}{2\#A},\ \forall\omega\in\tilde{\omega}^{s}.
\]

\noindent Since for all $\omega$ that satisfies $\frac{L_{\tau_{k}(\omega)}(\omega)}{L_{\tau_{k}(\omega)-1}(\omega)}-1<-\frac{\epsilon_{s}}{\#A}$
and $\frac{L_{\tau_{k}(\omega)-1}(\omega)}{L_{s}(\omega)}-1\leq\frac{\epsilon_{s}}{2\#A}$
we have

\noindent 
\[
\frac{L_{\tau_{k}(\omega)}(\omega)}{L_{\tau_{k}(\omega)-1}(\omega)}\frac{L_{\tau_{k}(\omega)-1}(\omega)}{L_{s}(\omega)}=\frac{L_{\tau_{k}(\omega)}(\omega)}{L_{s}(\omega)}<(1-\frac{\epsilon_{s}}{\#A})(1+\frac{\epsilon_{s}}{2\#A})=1+\frac{\epsilon_{s}}{2\#A}-\frac{\epsilon_{s}}{\#A}-\frac{\epsilon_{s}^{2}}{2(\#A)^{2}},
\]

\noindent it follows that

\noindent 
\begin{equation}
\frac{L_{\tau_{k}(\omega)}(\omega)}{L_{s}(\omega)}-1<-\frac{\epsilon_{s}}{2\#A}-\frac{\epsilon_{s}^{2}}{2(\#A)^{2}}.\label{eq:4}
\end{equation}

Combining \eqref{eq:3} and \eqref{eq:4} shows that 

\noindent 
\[
P(\{\omega:\frac{L_{\tau_{k}(\omega)}(\omega)}{L_{s}(\omega)}-1<-(\frac{\epsilon_{s}}{2\#A}+\frac{\epsilon_{s}^{2}}{2(\#A)^{2}})\}|\,\tilde{\omega}^{s}\cap\{rule\ 1\})\geq\frac{\epsilon_{s}}{\#A},
\]

\noindent where $-(\frac{\epsilon_{s}}{2\#A}+\frac{\epsilon_{s}^{2}}{2(\#A)^{2}})<-\frac{\epsilon_{s}}{2\#A}$
yields that 

\[
\begin{array}{l}
P(\{\omega:|\frac{L_{\tau_{k}(\omega)}(\omega)}{L_{s}(\omega)}-1|>\frac{\epsilon_{s}}{2\#A}\}|\,\tilde{\omega}^{s}\cap\{rule\ 1\})=\\
P(\{\omega:\frac{L_{\tau_{k}(\omega)}(\omega)}{L_{s}(\omega)}-1<-\frac{\epsilon_{s}}{2\#A}\}\cup\{\omega:(\frac{L_{\tau_{k}(\omega)}(\omega)}{L_{s}(\omega)}-1>\frac{\epsilon_{s}}{2\#A}\}|\,\tilde{\omega}^{s}\cap\{rule\ 1\})\geq\\
P(\{\omega:\frac{L_{\tau_{k}(\omega)}(\omega)}{L_{s}(\omega)}-1<-\frac{\epsilon_{s}}{2\#A}\}|\,\tilde{\omega}^{s}\cap\{rule\ 1\})\geq\frac{\epsilon_{s}}{\#A}>\frac{\epsilon_{s}}{2\#A}
\end{array}
\]

\noindent and the result follows.
\end{proof}
\vspace{0cm}

\vspace{0cm}

The proof of Theorem \ref{Th: relexing informativeness} is generalized
to the case where the number of elements, $|A|$, is arbitrary and
it is relied on achieving a uniform bound on the up-crossing probability
of any non-negative supermartingale which admits sufficiently (finite)
many fluctuations. 
\begin{proof}[\textbf{Proof of Theorem \ref{Th: relexing informativeness}}]
 Let $\epsilon,\psi,\nu\in(0,1)$. We will show that there exists
a uniform constant $K=K(\epsilon,\psi,\nu)$ such that on the set
of histories $\omega^{t}$ of $P^{H}-probability-(1-\epsilon)$, and
for all $n>0,$ only two scenarios are possible; if there exists a
subsequence of times $(k_{i})_{i=1}^{K+1}\subset\{1,...,n\}$, and
there exists a subsequence of corresponding outcomes $(a_{k_{i}})_{i=1}^{K+1}\subset A^{K+1}$
such that $|\rho(a_{k_{i}}|l_{k_{i}}(\omega),H)-\rho(a_{k_{i}}|l_{k_{i}}(\omega),L)|\geq\epsilon$
for all $1\leq i\leq K+1$, then, the value $l_{n}(\omega)$ is strictly
less than $\epsilon$, and more importantly, it remains below $\epsilon$
for all future periods $k$ from time $n$ onward. In all other scenarios,
the distance between the corresponding transitions satisfies $|\rho(a_{k}|l_{k}(\omega),H)-\rho(a_{k}|l_{k}(\omega),L)|\geq\epsilon$
for all $a_{k}\in A$ in all but $K$ periods $k$ in $\{1,...,n\}.$

As in the construction stated in \ref{subsec:Construction}, define
the increasing sequence of stopping times $\{\tau_{k}\}_{k=0}^{\infty}$
relative to $l:=\{l_{k}\}_{k\geq0}$ and $\epsilon$ inductively.
Let $\tilde{l}:=\{\tilde{l}_{k}\}_{k\geq0}$ be the result in faster
process. From Lemma \ref{Lemma: L is supermartingale} $l$ is a supermartingale;
hence from a standard result, $\tilde{l}$ is a supermartingale. Furthermore,
by Theorem \ref{th: l_k is weakly active supermartingale}, $\tilde{l}$
is weaky active supermartingale with activity $\frac{\psi}{2|A|}$
and rate $\nu.$

Applying Theorem \ref{Th: uniform rate for weak active supermartingale}
we are provided with an integer $K=K(\epsilon,\psi,\nu)>0$ (depending
only on these variables) such that for any weakly active supermartingale
$\tilde{l}$ with activity $\frac{\psi}{2|A|}$ rate $\nu,$ started
at $\tilde{l}_{0}\equiv0$ one has

\begin{equation}
P^{H}(\underset{k>K}{sup}\tilde{l}_{k}<\epsilon)>1-\epsilon.\label{eq:supL-1}
\end{equation}

In addition, By Lemma \ref{Lemma: distance implies probability},
whenever $|\rho(a|l_{k}(\omega),H)-\rho(a|l_{k}(\omega),L)|\geq\epsilon$
for some $k$ and $a\in A,$ then condition \eqref{condition1} holds.
Consequently, the process $\text{\ensuremath{\tilde{l}}}$ takes into
account all these observations and omits only observations where $|\rho(a|l_{k}(\omega),H)-\rho(a|l_{k}(\omega),L)|\leq\epsilon$
for all $a\in A$ (although, by condition \eqref{condition 2-1},
not necessarily all of them). 

As a result, conditional on the state being $H$ there exists a universal
constant $K=K(\epsilon,\psi,\nu)$, which does not depend on the pair
$(F^{L},F^{H})$, so that on the set of histories, $\omega^{t},$
of probability $(1-\epsilon)$ under $P^{H}$, in all but $K$ periods
either $|\rho(a|l_{k}(\omega),H)-\rho(a|l_{k}(\omega),L)|\leq\epsilon$
for all $a\in A$ or $l_{k}(\omega)<\epsilon.$ 

Now assume that there exist $\text{\ensuremath{K+1} periods }(k_{i})_{i=1}^{K+1}\subset\{1,...,n\}$
and $(a_{k_{i}})_{i=1}^{K+1}\subset A^{K+1}$ such that $|\rho(a_{k_{i}}|l_{k_{i}}(\omega),H)-\rho(a_{k_{i}}|l_{k_{i}}(\omega),L)|\geq\epsilon$
for all $1\leq i\leq K+1$, with $P^{H}(\omega^{k_{i}-1})>0$ and
let $n>K+1.$ Then inequality \eqref{eq:supL-1} ensures us that with
$P^{H}-probability-(1-\epsilon)$ 

\begin{equation}
\tilde{l}_{K+1}=l_{\tau_{K+1}}<\epsilon\label{eq:9-1}
\end{equation}

\noindent where by condition \eqref{eq:2} for any $k\geq n\geq\tau_{K+1}$
we obtain that either $\tilde{l}_{k}$ drops below $\epsilon$ or 

\begin{equation}
l_{k}(\omega)<l_{\tau_{K+1}}(\omega)(1+\frac{\epsilon}{2})<\epsilon(1+\epsilon)\label{eq:10-1}
\end{equation}

\noindent and hence it cannot exceed $\epsilon(1+\epsilon$). 

We conclude that there exists a constant $K$, which does not depend
on any pair $(F^{L},F^{H}),$ such that for any sufficiently large
$n>K$, with $P^{H}$ - probability - $(1-\epsilon);$ if there exist
$K+1$ periods in which $P^{H}$ and $P^{L}$ are slightly different
above $\frac{\psi}{k^{\nu}}$ in periods $k$ then the likelihood
ratio at any time after $n$ never exceeds $\epsilon(1+\epsilon).$ 

\vspace{0cm}
\end{proof}
\vspace{0cm}

\section{Efficiency}

\vspace{0.5cm}

\begin{proof}[\textbf{\small{}Proof of Theorem \ref{Th: finite expected time of the fist correct action},
part a}]

Let $\psi,\nu\in(0,1)$ and observe that for all $k\geq1$one has

\[
\begin{array}{l}
P^{H}(\tau=k)=(1-F^{H}(\bar{p}_{1}(\tilde{l}_{k})))\stackrel[t=1]{k-1}{\prod}F^{H}(\bar{p}_{1}(\tilde{l}_{t}))\\
\leq\stackrel[t=1]{k-1}{\prod}F^{H}(\bar{p}_{1}(\tilde{l}_{t}))\leq\stackrel[t=1]{k-1}{\prod}F^{L}(\bar{p}_{1}(\tilde{l}_{t}))-\frac{\psi}{(t+1)^{\nu}})\leq\stackrel[t=1]{k-1}{\prod}(1-\frac{\psi}{(t+1)^{\nu}}).
\end{array}
\]

\noindent In addition, since the exponential term $(1-\frac{\psi}{(k+1)^{\nu}})^{k}$
is dominated by the polynomial one $\frac{1}{k^{3}}$ for all sufficiently
large $k,$  there exist $N>0$ and $K:=K(\psi,\nu)=\stackrel[t=2]{N}{\sum}t(1-\frac{\psi}{t^{\nu}})^{t-1}+\stackrel[t=N+1]{\infty}{\sum}\frac{1}{t^{2}}$
such that 

\[
\begin{array}{l}
E^{H}[\tau]=\stackrel[t=1]{\infty}{\sum}tP^{H}(\tau=t)\leq\stackrel[t=2]{\infty}{\sum}t\stackrel[k=1]{t-1}{\prod}(1-\frac{\psi}{(k+1)^{\nu}})\leq\stackrel[t=2]{\infty}{\sum}t(1-\frac{\psi}{t^{\nu}})^{t-1}=\\
=\stackrel[t=2]{N}{\sum}t(1-\frac{\psi}{t^{\nu}})^{t-1}+\stackrel[t=N+1]{\infty}{\sum}t(1-\frac{\psi}{t^{\nu}})^{t-1}\leq\stackrel[t=2]{N}{\sum}t(1-\frac{\psi}{t^{\nu}})^{t-1}+\stackrel[t=N+1]{\infty}{\sum}\frac{1}{t^{2}}=K<\infty,
\end{array}
\]

\noindent and the result follows.
\end{proof}
\vspace{0cm}

\end{titlepage}
\end{document}